\documentclass{article}
\usepackage[utf8]{inputenc}
\usepackage{mathtools}
\usepackage{amsfonts} 
\usepackage{amssymb}
\usepackage{appendix}
\usepackage{amsmath}
\usepackage[pdftex]{graphicx}
\usepackage{tikz-cd}
\usepackage{mathrsfs}
\usepackage{titlesec}
 \usepackage{pstricks-add}
 \usepackage{comment}
 \usepackage{epsfig}
 \usepackage[nottoc,notlof,notlot]{tocbibind} 
\usepackage{pst-grad} 
\usepackage{pst-plot} 
\usepackage[space]{grffile} 
\usepackage{etoolbox} 
\usepackage{hyperref}
\usepackage{color}
\usepackage{amsthm}
\hypersetup{
    colorlinks,
    citecolor=blue,
    filecolor=black,
    linkcolor=purple,
    urlcolor=black
}
\newcommand{\VecF}[1]{\mathfrak{X}^{\infty}\left(#1\right)}

\newcommand{\limucp}{\mathop{\lim}\limits_{\substack{\mathrm{ucp} \\ \epsilon \rightarrow 0}}}

\newcommand{\hittingtime}[1]{\tau^h_{#1}}
\newcommand{\exittime}[1]{\tau^e_{#1}}
\newcommand{\hittingtimeproc}[2]{\tau^h_{#1}\left(#2\right)}
\newcommand{\exittimeproc}[2]{\tau^e_{#1}\left(#2\right)}
\newcommand{\stopproc}[2]{#1^{|#2}}

\newcommand{\Sint}[2]{\int #1\bullet d #2}

\newcommand{\dotp}[2]{\left\langle #1, #2 \right\rangle}

\newcommand{\pard}[2]{\frac{\partial#1}{\partial#2}}
\newcommand{\Ac}[2]{\mathcal{S}_{#1}(#2)}

\newcommand{\pr}[1]{\mathrm{pr}_{#1}}

\newcommand{\pb}[2]{\left\{#1, #2\right\}}

\newcommand{\cder}[2]{\frac{D}{D#1}#2}

\begin{document}
\title{Stochastic Implicit Lagrange-Poincaré Reduction}
\author{Archishman Saha}
\date{}
\def\tsph{\mathbb{S}^2}
\def\E{\mathbb{E}}
\def\pf{\noindent{\textbf{Proof.}\:}}
\def\expec{\mathbb{E}}
\def\Rp{\mathbb{R}_+}
\def\lag{\mathfrak{g}}
\def\alag{\tilde{\mathfrak{g}}}
\def\dlag{\mathfrak{g}^*}
\def\Ver{\mathrm{Ver\:}}
\def\Hor{\mathrm{Hor\:}}
\def\ver{\mathrm{Ver}}
\def\hor{\mathrm{Hor}}
\def\dalag{\tilde{\mathfrak{g}}^*}
\def\pontbundle{TQ\oplus T^*(Q)}
\def\redpontbundle{T(Q/G)\oplus T^*(Q/G)\oplus \alag\oplus\dalag}
\def\ker{\text{ker}}
\def\im{\text{im}}
\def\sthm{(M,\{\cdot,\cdot\},h,X)}
\def\hamsec{\tilde{X}_H}
\def\poin{Poincar\'e\:}
\def\pb{\{,\cdot,\}}
\def\Cinc{C^{\infty}_c}
\def\R{\mathbb{R}}
\def\yvec{\mathbf{y}}
\def\pman{(M, \{\cdot,\cdot\})}
\def\reduced{\mathrm{red}}
\def\sman{(M, \Omega)}
\def\Bvec{\mathbf{B}}
\def\Jvec{\mathbf{J}}
\def\cotman{(T^*M, \Omega)}
\def\inv{^{-1}}
\def\tman{(TM, \Omega)}
\def\Ito{It\^o\:}
\def\d{\mathbf{d}}
\def\D{\mathcal{D}}
\def\qvec{\mathbf{q}}
\def\pvec{\mathbf{p}}
\def\lvec{\mathbf{\Lambda}}
\def\deldeleps{\frac{\partial}{\partial \epsilon}\Big|_{\epsilon = 0}}
\def\Jvec{\mathbf{J}}
\def\Avec{\mathbf{A}}
\def\dd{\mathbf{d}d}
\def\cotM{T^*M}
\def\cotQ{T^*Q}
\def\L{\mathcal{L}}
\def\vlift{\mathrm{vlft}}
\def\Qvec{\mathbf{Q}}
\def\P{\mathcal{P}}
\def\Pvec{\mathbf{P}}
\def\Vvec{\mathbf{V}}
\def\uvec{\mathbf{u}}
\def\vvec{\mathbf{v}}
\def\Cin{C^{\infty}}
\def\F{\mathbb{F}}
\def\E{\mathbb{E}}
\def\stin{\int\left<\Theta,\D X\right>}
\def\omsec{\tilde{\Omega}}
\def\omsym{\Omega^{\text{sym}}}
\def\Sym{\textrm{Sym}^2}
\def\stpM{\tau_p M}
\def\kepphase{T^*(\mathbf{R}^3\setminus\{0\})}
\def\Ad{\mathrm{Ad}}
\def\ad{\mathrm{ad}}
\def\scpM{\tau^*_p M}
\def\del{\bullet d}
\def\Del{\bullet D}
\def\vpr{\mathrm{vpr}}
\def\pb{\{\cdot,\cdot\}}
\def\scM{\tau^* M}
\def\scN{\tau^* N}
\def\M{\mathcal{M}}
\def\scqN{\tau^*_q N}
\def\stM{\tau M}
\def\lco{L\'{a}zaro-Cam\'{\i}\:and\:Ortega\:}
\def\Deldeleps{\frac{D}{D\epsilon}\Big|_{\epsilon = 0}}
\def\Zvec{\mathbf{Z}}
\def\stN{\tau N}
\def\dim{\text{dim\:}}
\def\grad{\mathbf{\nabla}}
\def\covar{\delta^A}
\def\stqN{\tau_q N}
\def\grad{\mathbf{\nabla}}
\def\ostar{\textcircled{$\star$}}
\def\ostarthm{\textup{\textcircled{$\star$}}}
\def\xvec{\mathbf{x}}
\def\Xvec{\textbf{X}}
\def\TxS{T^\times S^3_{np}}
\def\ToneS{T^{\times}_1S^3_{np}}
\def\S{\mathcal{S}}
\newtheorem{thm}{Theorem}[section]
\newtheorem{cor}{Corollary}[section]
\newtheorem{lem}{Lemma}[section]
\newtheorem{defn}{Definition}[section]
\newtheorem{example}{Example}[section]
\newtheorem{prop}{Proposition}[section]
\newtheorem{rem}{Remark}[section]
\newtheorem{theorem}{Theorem}[section]
\newtheorem{remark}{Remark}[section]
\newtheorem{proposition}{Proposition}[section]
\newtheorem{lemma}{Lemma}[section]
\newtheorem{corollary}{Corollary}[section]
\newtheorem{definition}{Definition}[section]
\newtheorem{conjecture}{Conjecture}[section]
\newtheorem*{theorem*}{Theorem}
\maketitle
\allowdisplaybreaks
\begin{abstract}
    In this paper we consider reduction of the stochastic Hamilton-Pontryagin principle formulated on the Pontryagin bundle of a manifold $Q$. We prove that a stochastic action invariant under the free and proper action of a Lie group $G$ drops to a reduced variational principle expressed in terms of variables of the Pontryagin bundle of the reduced space $Q/G$, the associated adjoint bundle $\tilde{\mathfrak{g}}:= (Q\times \mathfrak{g})/G$ and its dual bundle $\tilde{\mathfrak{g}}^*$. This provides a stochastic analogue of deterministic implicit Lagrange-Poincaré reduction. The stochastic Euler-Lagrange equations drop to a set of stochastic horizontal and vertical Lagrange-Poincaré equations on $T(Q/G)\oplus T^*(Q/G)\oplus\tilde{\mathfrak{g}}\oplus\tilde{\mathfrak{g}}^*$. As examples, we consider stochastic perturbations of the rigid body with a rotor, as well as a Kaluza-Klein description of stochastic perturbations of a charged particle in a magnetic field.

\end{abstract}
\noindent\rule{\textwidth}{0.4pt}
\tableofcontents

\section{Introduction}

The goal of this paper is to study symmetries of stochastic variational principles, and extend the general framework of implicit Lagrange-Poincaré reduction carried out by Yoshimura and Marsden \cite{yoma} to the stochastic case. We also generalize stochastic Euler-Poincaré reduction carried out by Bou-Rabee and Owhadi \cite{Bou_Rabee_2008}, Arnaudon, de Castro and Holm \cite{holm2}, and Street and Takao \cite{street2023} to the case where the configuration manifold is an arbitrary smooth manifold equipped with the free and proper action of a Lie group. 

\medskip

The reduction of mechanical systems by symmetry forms a central part of the study of mechanics. In Hamiltonian reduction theory, the main idea is to show that geometric structures underlying the phase space of a mechanical system, for instance, a Poisson structure or a symplectic 2-form, pass to the quotient under reduction by a Lie group. At the level of dynamics, when the Hamiltonian is invariant under the Lie group action, the equations of motion in the unreduced space, and their reduced counterparts, are Hamiltonian with respect to the unreduced and reduced structures respectively. We refer the reader to Abraham and Marsden \cite{abraham2008foundations}, Marsden and Ratiu \cite{marsden2}, Marsden \cite{marsden1} and Holm, Schmah and Stoica \cite{gms} for expositions on this.

\medskip

In Lagrangian mechanics, the central role is played by the action and the variational principle determined by it. Lagrangian reduction focusses on how the action integral, and the associated equations for determining the critical points of it, drop to a reduced space under reduction by symmetry. While it is tied to Hamiltonian reduction via the Legendre transform, in practice, however, it is carried out independently of Hamiltonian reduction. For details of Lagrangian reduction, the reader may refer to \cite{marsden1}, \cite{marsden2}, \cite{gms}, Marsden and Scheurle \cite{marsden1993reduced, marsden1993lagrangian2}, and Cendra, Marsden and Ratiu \cite{cendra1997lagrangian}. 

\medskip

In 1901, Poincaré \cite{Poincare1901Action} introduced a `new form' of equations of motion on a Lie algebra. These equations are known as Euler-Poincaré equations and they arise by Lagrangian reduction by symmetry when the configuration space is Lie group. This is applicable to many important problems such as the free rigid body, and the Euler-Arnold equations for an ideal fluid (see Arnold and Khesin \cite{arnold1999topological} for details).

\medskip

The extension of Euler-Poincaré reduction to the case where the configuration manifold is a smooth manifold $Q$ acted upon by a Lie group $G$ is called \emph{Lagrange-Poincaré reduction}. In this case, the reduced space is the sum of tangent bundle of the shape space, $Q/G$, and the associated adjoint bundle obtained by fixing a connection on $Q\rightarrow Q/G$. The reduced equations may be divided into a set of horizontal equations, which resemble Euler-Lagrange equations with a forcing term, and vertical equations, which resemble Euler-Poincaré equations. Details on Lagrange-Poincaré reduction may be found in Cendra, Marsden and Ratiu \cite{cendra1997lagrangian}, Cendra, Marsden, Pekarsky and Ratiu \cite{cendra2003variational} and Yoshimura and Marsden \cite{yoma}. In \cite{yoma}, the authors show that Lagrangian, Hamiltonian, and the variational point of view of reduction can be accommodated in a single reduction procedure, called implicit Lagrange-Poincaré reduction. This is done by considering a constrained variational principle on the Pontryagin bundle $TQ\oplus \cotQ$ of $Q$ that accounts for the Legendre transform. The present paper extends this reduction procedure to stochastic variational principles.

\medskip

There are two different approaches to stochastic variational principles. The first approach considers an action functional involving the expectation of a Lagrangian, and replaces velocities by drifts of semimartingales. For details on this, we refer to Cipriano and Cruzeiro \cite{cipriano2007navier}, Arnaudon and Cruzeiro \cite{arnaudon2012lagrangian}, Arnaudon, Chen and Cruzeiro \cite{arn}, and Chen, Cruzeiro and Ratiu \cite{chen2023stochastic}. The latter two papers consider Euler-Poincaré reduction and semidirect product reduction respectively, in this setting. 

\medskip

The second approach, which we use in this paper, is the perturbative one, wherein a deterministic action functional is perturbed by incorporating stochastic terms to model the effects of random external noise on mechanical systems. Roughly speaking, the main motivation for this arises from the following question: suppose we have a way to measure some observable of the whole system, and assume that there is some random noise associated with our measurements. Can we consider a stochastic perturbation of the deterministic system such that the statistics of evolution of the given observable matches its measurement? This point of view has lead to several developments in data driven models of uncertainty in geophysical fluid dynamics and oceanography. See, for instance, Holm \cite{holm1}, Cotter, Gottwald, and Holm \cite{cotter2017stochastic}, Gay-Balmaz and Holm \cite{gay2018stochastic}, and Cotter et al. \cite{cotter2019numerically}. 

\medskip

To give a motivating finite dimensional example, consider the problem a satellite with a rotor attached to it. Suppose that the body angular momentum of the satellite and the angular momentum of the rotor can be observed, but that these observations are inherently noisy due to sensor error or environmental disturbances. These angular momenta are naturally expressed as variables on a reduced space obtained by exploiting the rotational symmetry of the system. We can then ask if these noisy equations in the reduced space can be obtained via reducing a stochastic perturbation of the entire system. This also leads to the problem of understanding when stochastic forcing preserves the symmetries underlying the reduction and when it breaks them.

\medskip

From a variational point of view, stochastic perturbations of mechanical systems can be modelled by the stochastic Hamilton-Pontryagin principle. In the stochastic Hamilton-Pontryagin principle we look for critical points of the following stochastic action defined on the Pontryagin bundle of manifold $Q$:
\begin{align*}
        \Ac{X}{\Gamma} = &\int_0^T\left(\L(q_t, v_t)\del X^0_t +\sum_{i = 1}^k L_i(q_t)\del X^i_t \right.\nonumber\\ &\left.+ \dotp{p_t}{\del q_t - v_t\del X^0_t - \sum_{i = 1}^k V_i(q_t)\del X^i_t}\right),
    \end{align*}
where $\del$ refers to Stratonovich integration. Here $\L\in \Cin(TQ)$, $L_1, \cdots, L_k\in \Cin(Q)$, $V_1, \cdots, V_k$ are vector fields on $Q$ and $X = (X^0, \cdots, X^k)$ is a semimartingale on $\R^{k+1}$. Typically, we are interested in the case when $X^0 = t$, in which case, we think of $L_i$ as the potential energy contribution of the external perturbation, and the vector fields $V_i$ as kinetic contributions by stochasticizing the relation $\dot{q} = v$. The critical points of $\mathcal{S}_X$ are determined by solving the stochastic implicit Euler-Lagrange equations given by  
\begin{align*}
    \del q_t &= v_t\del X^0_t + \sum_{i=1}^kV_i(q_t)\del X^i_t\nonumber\\
    \del p_t &= \frac{\partial}{\partial q_t}\left(\L\del X_t^0 + \sum_{i = 1}^k\left(L_i - \dotp{p_t}{V_i(q_t)}\right)\del X^i_t\right)\nonumber\\
    \left(p_t - \pard{\L}{v_t}\right)\del X^0_t&=0.
\end{align*}
We consider the case when $\L, L_i$ and $V_i$ are invariant under the free and proper action of a Lie group $G$. In this case the action is $G$-invariant, and once a connection $A$ has been fixed on the bundle $Q\rightarrow Q/G$, it drops to a reduced action on the bundle $TQ\oplus \cotQ\oplus\alag\oplus \dalag$, where $\alag = (Q\times \lag)/G$ is the associated bundle corresponding to $A$. The critical points of the reduced action, under certain stochastic constrained variations, correspond to two sets of equations, namely, the stochastic implicit horizontal Lagrange-Poincaré equations and the stochastic implicit vertical Lagrange-Poincaré equations. These equations can be described in a coordinate invariant manner by using the notion of the stochastic covariant derivative described by Norris \cite{norris1992complete} and Catuogno, Ledesma and Ruffino \cite{catuogno2013note}. 

\medskip

The main contributions of this paper are as follows:
\begin{enumerate}
    \item We provide a stochastic extension of implicit Lagrange–Poincaré reduction, showing that invariance of the stochastic action yields reduced equations that naturally incorporate curvature and noise.
    \item We use the techniques of covariant Stratonovich calculus to ensure that the reduced equations are coordinate independent. 
    \item As examples of the general theory, we show how stochastic perturbations of familiar deterministic systems, such as the rigid body with a rotor, and the Kaluza-Klein approach to charged particles in a magnetic field, fit into our framework.
\end{enumerate}

The paper is divided as follows: in Section 2, we review deterministic Lagrange-Poincaré reduction. In Section 3 we describe variations of semimartingales on manifolds and state the stochastic Hamilton-Pontryagin principle. Section 4 is dedicated to the description of the stochastic covariant derivative on a vector bundle. Section 5 contains details of stochastic Lagrange-Poincaré reduction, and examples are discussed in Section 6.

\section{Implicit Lagrange-Poincaré Reduction}

The reader is referred to Yoshimura and Marsden \cite{yoma} and Cendra, Marsden and Ratiu \cite{cendra1997lagrangian} for details on implicit Lagrange-Poincaré reduction. In this section we will provide a summary of their results as well as introduce relevant notations.

\subsection{The Hamilton-Pontryagin Principle}

Let $Q$ be a configuration manifold and $\P Q:= TQ\oplus \cotQ$ denote its Pontryagin bundle. Given a Lagrangian $\L\in\Cin(TQ)$ and a curve $(q(t), v(t), p(t))$ in $\P Q$, the \textbf{Hamilton-Pontryagin action} is defined by 
\begin{equation}\label{eq:deterministic_Hamilton_Pontryagin}
    \S (q(t), v(t), p(t)):= \int_0^T [\L(q(t), v(t)) + \dotp{p(t)}{\dot{q}(t) - v(t)}]dt.
\end{equation}
The following theorem concerns the characterization of critical points of $\S$ under fixed endpoint variations of $q(t)$:
\begin{theorem}\label{thm:deterministic_EL}
    A curve $(q(t), v(t), p(t))$ in $\P Q$ is a critical point for $\S$ for deformations $\epsilon\mapsto (q_{\epsilon}(t), v_{\epsilon}(t), p_{\epsilon}(t))$ such that $\delta q(t) = 0$ at $t = 0$ and $t = T$ if and only if $(q(t), v(t), p(t))$ satisfies the \textbf{implicit Euler-Lagrange equations} given by
    \begin{equation}\label{eq:deterministic_implicit_EL}
        \dot{q} = v,\:\:\:\dot{p} = \frac{\partial \L}{\partial q},\:\:\:p = \frac{\partial \L}{\partial v}.
    \end{equation}
\end{theorem}
For an application of the Hamilton-Pontryagin principle to constrained systems, the reader is referred to Yoshimura and Marsden \cite{yoshimura06} and \cite{yoshimura061}.

\subsection{The Geometric Setup of Implicit Lagrange-Poincaré Reduction}

Let $\phi:G\times Q\rightarrow Q$ be a free and proper action of a Lie group $G$ on $Q$. Given $g\in G$ and $q\in Q$, we will denote by $gq$, or $\phi_g(q)$ the element $\phi(g,q)\in Q$. If $v_q\in T_qQ$ and $p_q \in T^*_qQ$, we let $gv_q := T_q\phi_g(v_q)$ and $gp_q = T^*_{gq}\phi_{g\inv}(p_q)$. Let $A:TQ\rightarrow\lag$ denote a principal connection on the principal bundle $\pi:Q\rightarrow Q/G$. For each $q\in Q$, the equivalence class $\pi(q)$ will be denoted by $[q]$ or $[q]_G$. The horizontal and vertical subspaces at $q$ are denoted by $\hor_q$ and $\ver_q$ respectively, whereas $\Hor TQ$ and $\Ver TQ$ will denote the corresponding horizontal and vertical bundles.
\begin{lemma}{\cite[Lemma 2.2.1]{cendra1997lagrangian}}\label{lem:expression_for_A_in_terms_of_g}
    Let $q(t)$ be a curve in $Q$, $q^h(t)$ the horizontal lift of $\pi(q(t))$ starting at $q_0 = q(0)$, and $g^q(t)$ a $G$-valued curve satisfying $q(t) = g^q(t)q^h(t)$. Then $A(q(t), \dot{q}(t)) = \dot{g}^q(t){g^q(t)}\inv$. 
\end{lemma}
The curvature of $A$ is the $\lag$-valued 2-form $B$ defined by $B(X,Y) = \mathbf{d}A(\Hor X, \Hor Y)$, where $\mathbf{d}A$ is the covariant exterior derivative of $A$ and $X$ and $Y$ are vector fields on $Q$. This satisfies
\begin{equation}\label{eq:curvatureformulas}
    B(X, Y) = -A([\Hor X, \Hor Y]) = \d A(X, Y) - [A(X), A(Y)].
\end{equation}
Let $G$ act on $Q\times \lag$ by $(q, \xi)\mapsto (gq, \Ad_g\xi)$ for all $g\in G$ and $\alag = (Q\times \lag)/G$. The equivalence class of $(q,\zeta)$ will be denoted by $[q, \zeta]_G$. We consider the associated vector bundle $\pi_{\alag}:(Q\times \lag)/G\rightarrow Q/G$ with fibre $\lag$. There is a natural Lie algebra structure on the fibres of $\alag$ given by $[[q, \eta]_G, [q, \zeta]_G] = [q,[\eta, \zeta]_G]$. By using $\Ad$-invariance of the Lie bracket, one can show that this is well defined. The covariant derivative of a curve $[q(t), \zeta(t)]_G$ in $\alag$ is given 
\begin{equation}\label{eq:covariantderivativeassociatedbundle}
    \cder{t}{[q(t),\zeta(t)]_G} = [q(t), [\zeta(t),A(q(t), \dot{q}(t)) ]+ \dot{\zeta}(t)]_G.
\end{equation}
The corresponding connection on $\alag$ is denoted by $\nabla^{\alag}$. Moreover, $\nabla^{\alag}$ defines a connection $\nabla^{\dalag}$ on the dual bundle $\pi_{\dalag}:\dalag\rightarrow Q/G$ in the following way: let $\bar\mu(t)$ and $\bar\zeta(t)$ be curves in $\dalag$ and $\alag$ respectively such that they project to the same curve in $Q/G$. The connection $\nabla^{\dalag}$ is the unique one such that the covariant derivative $\frac{D\bar{\mu}}{Dt}$ satisfies
\begin{equation}\label{eq:covariantderivativedualbundle}
    \frac{d}{dt}\dotp{\bar\mu(t)}{\bar\zeta(t)} = \dotp{\bar\mu(t)}{\frac{D\bar{\zeta}(t)}{Dt}} + \dotp{\frac{D\bar{\mu}(t)}{Dt}}{\bar\zeta(t)}.
\end{equation}
We have a well-defined isomorphism of fibre bundles $\Psi_A:TQ/G\rightarrow T(Q/G)\oplus \alag$ given by \[[v_q]_G\in TQ/G \mapsto (T_q\pi (v_q), [q, A_q(v_q)]_G)\in T(Q/G)\oplus \alag.\]The reader may refer to Cendra, Marsden and Ratiu \cite{cendra1997lagrangian} for its use in Lagrange-Poincaré equations. Dually, one has an isomorphism ${(\Psi_A\inv)}^*: \cotQ/G\rightarrow T^*(Q/G)\oplus\dalag$ given by \[\dotp{{(\Psi_A\inv)}^*([\alpha_q])}{({u_{[q]}},[q,\eta]_G)} = \dotp{(\alpha_q)_{q}^{h^*}}{u_{[q]}} + \dotp{J(\alpha_q)}{\eta}, \]where $({u_{[q]}},[q,\eta]_G)\in T(Q/G)\oplus \alag$, $J:\cotQ\rightarrow\dlag$ is the momentum map of the cotangent lifted $G$-action on $\cotQ$, $(\cdot)_q^h: T_qQ\rightarrow T_{[q]}(Q/G)$ is the horizontal lift map that sends $v_{[q]}\in T_{[q]}(Q/G)$ to $(T_q\pi|_{\hor_q})\inv(v_{[q]})$ and $(\cdot)_q^{h^*}$ is the dual of $(\cdot)_q^h$. Using $\Psi_A$ and ${(\Psi_A\inv)}^*$ we obtain an isomorphism \[\tilde{\Psi}_A:\P Q/G \rightarrow \redpontbundle = \P (Q/G)\oplus \alag\oplus \dalag.\]
\begin{remark}
    The bundle $TQ/G$ over $Q/G$ is called the \textbf{Atiyah quotient} (see Mackenzie \cite{mackenzie2005general}) while $T^*(Q/G)\oplus\dalag$ is called the \textbf{Weinstein space} (see Ortega and Ratiu \cite{ortega2013momentum}). 
\end{remark}
Corresponding to the curvature 2-form $B$, we can define a $\alag$-valued curvature 2-form on $Q/G$ given by 
\begin{equation}\label{eq:reduced_curvature}
    \dotp{\bar{\mu}}{\tilde{B}(\pi(q))(u_{[q]} v_{[q]})} = \dotp{\bar{\mu}}{[q,B(q)(u_q, v_q)]_G},
\end{equation}
where $q \in G$, $\mu\in \dlag$, $\bar{\mu} = [q, \mu]_G\in \dalag$, and $u_{[q]}$ and $v_{[q]}$ are elements of $T_{\pi(q)}(Q/G)$ with $u_{[q]} = T_q \pi(u_q)$ and $v_{[q]} = T_{q}\pi(v_q)$. Equivariance of $B$ can be used to show that $\tilde{B}$ is well-defined.

\subsection{The Structure of Variations}
Let $q(t)$ be a curve in $Q$ and $\epsilon\mapsto q_{\epsilon}(t)$ be a deformation of $q(t)$. Yoshimura and Marsden \cite{yoma} prove the following results:
\begin{theorem}\label{thm:deterministic_variations}
    Let $\xi(t) = A(q(t), \dot{q}(t))$. Then the following holds:
    \begin{enumerate}
        \item If $\delta q$ is vertical then 
        \begin{equation}\label{eq:vertical_variations}
            \delta \xi = \dot{\eta} + [\eta, \xi],
        \end{equation}
        where $\eta(t) = A(q(t), \delta q(t))$. If $\delta q(t)$ vanishes at $t = 0$ and $t = T$ then the same holds for $\eta(t)$.
        \item If $\delta q$ is horizontal then 
        \begin{equation}\label{eq:horizontal_variation}
            \delta \xi = B(q)(\delta q, \dot{q}).
        \end{equation}
    \end{enumerate}
\end{theorem}
Let $\bar{\xi}(t) := [q(t), \xi(t)]_G$. Given a deformation $\epsilon\mapsto q_{\epsilon}(t)$, we define the covariant variation of $\bar{\xi}$ by 
\begin{equation}\label{eq:covariant_variation}
    \delta^A\bar{\xi} = \Deldeleps [q_{\epsilon}(t), \xi_{\epsilon}(t)]_G.
\end{equation}
As a consequence of Theorem \ref{thm:deterministic_variations} and the definition of the covariant derivative on the adjoint bundle, it can be shown that 
\begin{equation}\label{eq:total_covariant_variation}
    \covar\bar{\xi} = \frac{D\bar{\eta}}{Dt} + [\bar{\xi}, \bar{\eta}] + \tilde{B}(\delta x, \dot{x}),
\end{equation}
where $\bar{\eta} = [q, A(q, \delta q)]_G$. If $\delta q(t)$ vanishes at $t = 0$ and $t = T$ then so does $\bar{\eta}(t)$.

\medskip

Let $(q(t), v(t), p(t))$ be a curve in $\P Q$ and consider the projection $[q(t), v(t), p(t)]_G$ on the reduced Pontryagin bundle $\P Q/G$. We let $\tilde{\Psi}_A([q(t), v(t), p(t)]_G) = (x(t), u(t), y(t), \bar{\zeta}(t),\bar{\mu}(t))$. Here $x(t) = \pi(q(t))$, $u(t) = T_{q(t)}\pi(v(t))$, $y(t) = (p(t))_q^{h^*}$, $\bar{\zeta}(t) = [q(t), A(q(t), v(t))]_G$ and $\bar{\mu}(t) = [q(t), J(q(t), p(t))]_G$. If $(\delta q(t), \delta v(t),\delta p(t))$ is a general variation of $(q(t), v(t), p(t))$ then the variation of the curve $(x(t), u(t), y(t), \bar{\zeta}(t),\bar{\mu}(t))$ satisfies 
\[(\delta x(t),\delta u(t),\delta y(t),\delta \bar{\zeta}(t),\delta \bar{\mu}(t)) \in T_{x(t)}(Q/G)\oplus T_{u(t)}T(Q/G) \oplus T_{y(t)}T^*(Q/G)\oplus T_{\bar{\eta}(t)}\alag\oplus T_{\bar{\mu}(t)}\dalag.\]
\begin{remark}\label{rem:on_variations}
    As explained in \cite{yoma} and \cite{yoshimura06}, given a curve $(q(t), \dot{q}(t), p(t))$ with $\bar{\xi} := [q(t), A(q(t), \dot{q}(t))]_G$ we will only consider deformations $\epsilon\mapsto \bar{\xi}_{\epsilon}(t)$  such that $\pi_{\alag}(\bar{\xi}_{\epsilon}(t)) := x_{\epsilon}(t)$ does not depend on $\epsilon$. The variations corresponding to such deformations are called $\alag$-fiber variations, and $\delta \bar{\xi}(t)$ can be identified with an element in $\alag$ instead of $T_{\bar{\xi}(t)}\alag$. The variation $\delta^A\bar{\xi}(t) = \Deldeleps \bar{\xi}_{\epsilon}(t)$ is an instance of a $\alag$-fiber variation and is an element of the fiber $\alag_{x(t)}\cong \lag$. If $x_{\epsilon}(t)\oplus \bar{\xi}_{\epsilon}(t)$ is a family of curves in $(Q/G)\oplus\alag$ depending smoothly on $\epsilon$ then its \textbf{covariant variation} is given by
    \[\delta x(t) \oplus \delta ^A \bar{\xi}(t):= \deldeleps x_{\epsilon}(t) \oplus \Deldeleps \bar{\xi}_{\epsilon}(t).\]

    \medskip

    In the context of implicit Lagrange-Poincaré reduction, we are interested in the following example of covariant variation. Suppose $q(t)$ is a curve in $Q$ and $\epsilon\mapsto q_{\epsilon}(t)$ is a deformation of $q(t)$. This induces a variation of the curve $x(t)\oplus \bar{\xi}(t):= \pi(q(t))\oplus [q(t), A(q(t), \dot{q}(t))]_G$ given by \[\epsilon\mapsto x_{\epsilon}(t)\oplus \bar{\xi}_{\epsilon}(t):=  \pi(q_{\epsilon}(t))\oplus [q_{\epsilon}(t), A(q_{\epsilon}(t), \dot{q}_{\epsilon}(t))]_G.\]The covariant variation in this case is given by $\delta x\oplus \delta^A \bar{\xi}$, where \[\delta^A \bar{\xi} = \frac{D\bar{\eta}}{Dt} + [\bar{\xi}, \bar{\eta}] + \tilde{B}(\delta x, \dot{x})\] with $\bar{\eta} = [q, A(q, \delta q)]_G$. It is under these constrained variations that we will consider stationary points of the reduced action principle.
\end{remark}

\subsection{The Implicit Lagrange-Poincaré Reduction Theorem}

Let $\Pr_{TQ/G}:TQ\rightarrow TQ/G$, $\Pr_{\cotQ/G}:\cotQ\rightarrow \cotQ/G$ and $\Pr_{\P Q/G}: \P Q\rightarrow \P Q/G$ denote the corresponding projections onto equivalence classes. Suppose $\L\in\Cin(TQ)$ is invariant under the tangent lifted action of $G$ on $TQ$. The isomorphism $\Psi_A$ between $TQ/G$ and $T(Q/G)\oplus \alag$ yields a reduced Lagrangian $\ell\in \Cin(T(Q/G)\oplus \alag)$ such that $\L = \ell\circ \Psi_A\circ \Pr_{TQ/G}$. Then, we have a reduced action functional corresponding to $\S$ on $\P (Q/G) \oplus \alag \oplus \dalag$ given by 
\begin{align}\label{eq:deterministic_reduced_action}
    \S^{red}(x(t), u(t), y(t), \bar{\zeta}(t), \bar{\mu}(t)) &= \int_0^T \left[\ell(x(t), u(t), \bar{\zeta}(t))dt + \dotp{y(t)}{\dot{x}(t) - u(t)}\right.\nonumber\\
    &+\left.\dotp{\bar{\mu}(t)}{\bar{\xi}(t) - \bar{\zeta}(t)}\right].
\end{align}
The next theorem relates the critical points of $\S$ and $\S^{red}$. The proof of it can be found in \cite{yoma}.
\begin{theorem}[Implicit Lagrange-Poincaré Reduction Theorem]\label{thm:deterministic_LP_reduction}
The following are equivalent:
\begin{enumerate}
    \item The curve $(q(t), v(t), p(t))$ is a critical point of the Hamilton-Pontryagin action $\S$ for all variations $\delta q$, $\delta v$ and $\delta p$ with $\delta q(0) = \delta q(T) = 0.$
    \item The curve $(q(t), v(t), p(t))$ satisfies the implicit Euler-Lagrange equations \eqref{eq:deterministic_implicit_EL}.
    \item Let $[q(t), v(t), p(t)]_G$ denote the reduced curve and $(x(t), u(t), y(t), \bar{\zeta}(t), \bar\mu(t))$ denote the curve $\tilde{\Psi}_A([q(t), v(t), p(t)]_G)$. Then $(x(t), u(t), y(t), \bar{\zeta}(t), \bar\mu(t))$ is a critical point of the reduced action $\S^{red}$ for arbitrary variations $\delta u$, $\delta v$, $\delta y$ and $\delta \bar\mu$ and for variation $\delta x\oplus \covar\bar\xi$ such that $\delta x(0) = \delta x(T) = 0$ and 
    \[\covar\bar{\xi} = \frac{D\bar{\zeta}}{Dt} + [\bar{\xi}, \bar{\zeta}] + \tilde{B}(\delta x, \dot{x}),\]where $\bar\zeta(t)$ is an arbitrary curve in $\alag$ with $\bar\zeta(0) = \bar\zeta(T) = 0$. 
    \item The curve $(x(t), u(t), y(t), \bar{\zeta}(t), \bar\mu(t))$ satisfies the \textbf{horizontal Lagrange-Poincaré equations}
    \begin{equation}\label{eq:deterministic_horizontal_LP}
        \frac{Dy}{Dt} = \frac{\partial \ell}{\partial x} - \dotp{\bar{\mu}}{i_{\dot{x}}\tilde{B}},\:\:\dot{x} = u,\:\:y = \frac{\partial \ell}{\partial u},
    \end{equation}
    and the \textbf{vertical Lagrange-Poincaré equations}
    \begin{equation}\label{eq:deterministic_vertical_LP}
        \frac{D\bar\mu}{Dt} = \ad^*_{\bar{\xi}}\bar{\mu},\:\:\bar{\xi} = \bar{\zeta},\:\:\bar\mu = \frac{\partial \ell}{\partial \bar\zeta}.
    \end{equation}
\end{enumerate}
\end{theorem}
\begin{remark}
    The covariant derivative $\frac{Dy}{Dt}$ can be defined by fixing an affine connection $\nabla^{T(Q/G)}$ on the manifold $Q/G$. Then the definition of $\frac{Dy}{Dt}$ is analogous to that of $\frac{D\bar{\mu}}{Dt}$ (see Equation \eqref{eq:covariantderivativedualbundle}).
\end{remark}
\begin{remark}
    In the previous theorem, the partial derivatives of $\ell$ are interpreted as follows: the partial derivatives $\frac{\partial \ell}{\partial u}$ and $\frac{\partial \ell}{\partial \bar{\zeta}}$ are defined by:
    \begin{align}\label{eq:partial_derivatives_fiber}
        \dotp{\frac{\partial \ell}{\partial u}}{v} &= \frac{d}{ds}\Big|_{s = 0}\ell(x,u + sv,\bar{\zeta})\nonumber\\
        \dotp{\frac{\partial \ell}{\partial \bar{\zeta}}}{\bar{\eta}} &= \frac{d}{ds}\Big|_{s = 0}\ell(x,u,\bar{\zeta} + s\bar{\eta}),
    \end{align}
    where $v$ and $\bar{\zeta}$ are arbitrary elements of $T_x(Q/G)$ and $\pi_{\alag}^{-1}(x)$ respectively. To define the partial derivative $\frac{\partial l}{\partial x}$, we use the connection $\nabla^{T(Q/G)}$ on the $Q/G$. Then $\nabla^{T(Q/G)}\oplus \nabla^{\alag}$ defines an affine connection on the bundle $\nabla^{T(Q/G)}\oplus \nabla^{\alag}$. Let $x(t)$ be a curve in $Q/G$ with $x(0) = x_0$ and $(x(t), u^h(t), \bar{\zeta}^h(t))$ denote the horizontal lift of $x(t)$ with respect to the connection $\nabla^{T(Q/G)}\oplus \nabla^{\alag}$ with $(x(0), u^h(0), \bar{\zeta}^h(0)) = (x_0, u_0, \bar{\zeta}_0) \in T(Q/G)\oplus \alag$. We define 
    \begin{equation}\label{eq:partial_derivatives_x}
        \dotp{\frac{\partial \ell}{\partial x}\Big|_{(x_0, u_0, \bar\zeta_0)}}{(x(0), \dot{x}(0))} = \frac{d}{dt}\Big|_{t = 0}\ell(x(t), u^h(t), \bar{\zeta}^h(t)).
    \end{equation}
\end{remark}

\section{The Stochastic Hamilton-Pontryagin Principle}

The stochastic Hamilton-Pontryagin principle has been explored by Bou-Rabee and Owhadi \cite{Bou_Rabee_2008}, Street and Takao \cite{street2023}, Saha \cite{saha2025stochastic}, and Li, Gay-Balmaz, Shi and Wang \cite{li2025variational}. In this section, following \cite{saha2025stochastic}, we introduce variations of semimartingales on manifolds and describe the stochastic Hamilton-Pontryagin principle.

\subsection{Deformations of Manifold-Valued Semimartingale}

We will always consider continuous semimartingales defined on a probability space $(\Omega, \mathcal{F}, P)$. Given predictable stopping times $S$ and $T$, we define 
\begin{align*}
    [[S,T]] &= \{(\omega, t)\in \Omega\times[0, \infty)\:|\:S(\omega)\leq t\leq T(\omega)\}\\
    [[S,T[[ &= \{(\omega, t)\in \Omega\times[0, \infty)\:|\:S(\omega)\leq t< T(\omega)\}\\
    ]]S,T]] &= \{(\omega, t)\in \Omega\times[0, \infty)\:|\:S(\omega)<t\leq T(\omega)\}\\
    ]]S,T[[ &= \{(\omega, t)\in \Omega\times[0, \infty)\:|\:S(\omega)< t< T(\omega)\}.
\end{align*}Given a stopping time $T$, we let $\stopproc{\Gamma}{T}$ denote the stopped process $\Gamma^{|T}_t = \Gamma_{t\wedge T}$. If $M$ is a smooth manifold we denote by $\mathscr{S}(M)$ the space of (continuous) semimartingales on $M$ and suppose $\Gamma\in \mathscr{S}(M)$. Usually we will assume that $M$ is connected.
\begin{definition}
    A \textbf{deformation} of $\Gamma$ is a map $\epsilon\in (-s,s)\mapsto \Gamma_{\epsilon}\in \mathscr{S}(M)$ satisfying the following conditions:
    \begin{itemize}
        \item $\Gamma_{0,t} = \Gamma_t$.
        \item Given any $f\in \Cin(M)$ there exists a $TM$-valued semimartingale $\delta\Gamma$ over $\Gamma$ such that $\frac{f(\Gamma_{\epsilon}) - f(\Gamma)}{\epsilon}$ converges to $df(\delta\Gamma)$ in the semimartingale topology as $\epsilon \rightarrow 0$.
    \end{itemize}
    Given a deformation $\epsilon\mapsto \Gamma_{\epsilon}$ of $\Gamma$, the corresponding $TM$-valued semimartingale $\delta \Gamma$ is called the \textbf{variation} of $\Gamma$.
\end{definition}
\begin{remark}
    Following Definition 2.7 and 2.9 in Arnaudon and Thalmaier \cite{arnaudon1998}, it follows that the map $\epsilon \mapsto \Gamma_{\epsilon}$ is differentiable at $\epsilon = 0$ with respect to the semimartingale topology on $M$.
\end{remark}
\begin{definition}
    We say that $\Gamma\in\mathscr{S}(M)$ is \textbf{admissible} if, given any $Y\in\mathscr{S}(TM)$ over $\Gamma$, there exists a deformation $\epsilon\mapsto \Gamma_{\epsilon}$ of $\Gamma$ with $\delta \Gamma = Y$.
\end{definition}
\begin{remark}
    If $M$ is compact then, given any semimartingale $Y$ in $TM$ over $\Gamma$, there exists a deformation $\epsilon\mapsto \Gamma_{\epsilon}$ of $\Gamma$ such that $\delta \Gamma = Y$. This is a consequence of Corollary 4.3 in Arnaudon and Thalmaier \cite{arnaudon1998}. 
\end{remark}
Next, we describe how to construct variations of $\Gamma$ that vanish at time $t = 0$, $t = T$, as well as at the entry and exit times for a chart. Let $K$ be a closed subset of $M$. Let $\hittingtime{K}(\Gamma)$ denote the hitting time of $\Gamma$ for $K$ and $\exittime{K}(\Gamma)$ denote the exit time of $\Gamma$ from $K$. If $\tau^{(h,e)}_K:= \exittimeproc{K}{\Gamma_{t+\hittingtimeproc{K}{\Gamma}}}$ then, on the event\[[[\hittingtime{K}(\Gamma), \hittingtime{K}(\Gamma)+ \tau^{(h,e)}_K]]:= \{(\omega, t)\in \Omega\times[0, \infty)\:|\:\hittingtime{K}(\Gamma)\leq t\leq \hittingtime{K}(\Gamma)+ \tau^{(h,e)}_K\}\]$\Gamma$ takes its values in $K$. Now consider a vector field $X$ that is supported on the interior of $K$ and a smooth function $g\in\Cin(\R)$ that is supported on $(0,T)$. Let $Y_t:= g(t)X(\Gamma_t)$. Then $Y$ is $TM$-valued semimartingale over $\Gamma$ such that $Y\equiv 0$ in $[[0, \hittingtime{K}(\Gamma)]]\bigcup [[\left(\hittingtime{K}(\Gamma)+\tau^{(h,e)}_K\right)\wedge T,\infty[[$. The admissibility hypothesis shows that there exists a deformation $\epsilon\mapsto \Gamma_{\epsilon}$ of $\Gamma$ with $\delta \Gamma = Y$. Deformations of $\epsilon\mapsto \Gamma_{\epsilon}$ of $\Gamma$ such that the associated variations vanish in $[[0, \hittingtime{K}(\Gamma)]]\bigcup [[\left(\hittingtime{K}(\Gamma)+\tau^{(h,e)}_K\right)\wedge T,\infty[[$ are called $(K,T)$-\textbf{deformations} and their associated variations are called $(K,T)$-\textbf{variations}. 

\medskip

The next lemma, proven in \cite{saha2025stochastic}, is a stochastic version of the fundamental lemma of the calculus of variations:
\begin{lemma}\label{lem:stochastic_fundlem}
    Let $M$ be a smooth $n$-manifold and $U\subseteq M$ be a coordinate chart. We identify $U$ with an open subset of $\R^n$, also denoted by $U$. Let $\Gamma$ be an admissible semimartingale and $\Xi$ be a map between semimartingales on $M$ and semimartingales on $\R^n$ satisfying $\Xi(\Gamma)_{A_t} = \Xi(\Gamma_{A_t})$ for any continuous change of time $t\mapsto A_t$. If for every $(\bar{U},T)$-deformation $\epsilon\mapsto \Gamma_{\epsilon}$ we have
\[\int_{\hittingtime{\bar{U}}(\Gamma)}^{\hittingtime{\bar{U}}(\Gamma)+\tau^{(h,e)}_{\bar{U}}}\dotp{\delta\Gamma^{|T}}{\del \Xi(\stopproc{\Gamma}{T})} = 0,\]
    then $\del\Xi(\stopproc{\Gamma}{T}) = 0$ in $]]\hittingtime{U}, \hittingtime{U}+\tau^{(h,e)}_{U}[[$. Here $\del\Xi(\stopproc{\Gamma}{T}) = 0$ means that $\Xi(\Gamma^{|T}) - \Xi(\Gamma^{|T})_{\hittingtime{U}} = 0$ almost surely in $]]\hittingtime{U}(\Gamma), \hittingtime{U}(\Gamma)+\tau^{(h,e)}_{U}[[$.
\end{lemma}
\subsection{The Stochastic Implicit Euler-Lagrange Equations}
For defining variations of semimartingales, we considered differentiability in the semimartingale topology. On the other hand, the variation of action functionals is described by differentiability in the topology of uniform convergence on compact sets in probability (ucp).
\begin{definition}
    Let $\Gamma\in\mathscr{S}(M)$ be admissible and $\epsilon\mapsto \Gamma_{\epsilon}$ be a deformation of $\Gamma$. Given a 1-form $\alpha$ on $M$ and a smooth function $f\in \Cin(M)$, we define\begin{enumerate}
        \item $\D\Sint{f(\Gamma)}{X} = \limucp \Sint{\frac{f(\Gamma_{\epsilon}) - f(\Gamma)}{\epsilon}}{X_t}$.
        \item $\D\Sint{\alpha}{\Gamma} = \limucp \frac{1}{\epsilon}\left(\Sint{\alpha}{\Gamma_{\epsilon}} - \Sint{\alpha}{\Gamma}\right)$,
\end{enumerate}where the convergence is taken in the ucp topology.
\end{definition}

\begin{lemma}\label{lem:variationsofsemimartginales}
Let $\Gamma\in\mathscr{S}(M)$ and $\epsilon\mapsto \Gamma_{\epsilon}$ be a deformation of $\Gamma$.
\begin{enumerate}
    \item For every real semimartingale $X$ and $f\in \Cin(M)$
    \begin{equation}\label{eq:variationoffdx}
        \D\int f(\Gamma)\del X = \int df(\delta \Gamma)\del X
    \end{equation}
    \item For every 1-form $\alpha$ on $M$
    \begin{equation}\label{eq:variationofalphadgamma}
        \D\int \alpha(\Gamma)\del \Gamma = \int i_{\delta \Gamma}\d\alpha \del X + \langle\alpha(\Gamma),\delta \Gamma\rangle - \langle\alpha(\Gamma_0), \delta\Gamma_0 \rangle,
    \end{equation}
    where $\d\alpha$ denotes the exterior derivative of $\alpha$.
\end{enumerate}
\end{lemma}
The next result may be obtained either by purely deterministic arguments or as a corollary of the previous lemma:
\begin{corollary}\label{cor:importantcorollary}
    Let $\gamma(t)$ be a smooth curve in $M$ and $\alpha$ be a 1-form on $M$. Suppose $\epsilon\mapsto \gamma_{\epsilon}(t)$ be a variation of $\gamma(t)$. Then
    \begin{equation}
        \delta\dotp{\alpha(\gamma(t))}{\dot{\gamma}(t)} = i_{\delta \gamma}\d \alpha (\dot{\gamma}(t))+ \frac{d}{dt}\dotp{\alpha(\gamma(t))}{\delta \gamma}.
    \end{equation}
\end{corollary}
Next, we define the stochastic action that we will consider. 
\begin{definition}
Let $\L\in\Cin(TQ)$ and consider smooth functions $L_1, \cdots, L_k\in \Cin(Q)$ (thought of as stochastic potentials) and vector fields $V_1, \cdots, V_k$ on $Q$ (thought of as the contribution of the noise to the velocity). Given $k+1$ semimartingales $X^0, \cdots, X^k$ on $\R$ and a $\P Q$-valued semimartingale $\Gamma_t = (q_t, v_t, p_t)$ we define the \textbf{stochastic Hamilton-Pontryagin action integral} as 
    \begin{align}\label{eq:stochastic_Hamilton-Pontryagin_action}
        \Ac{X}{\Gamma} = &\int_0^T\left(\L(q_t, v_t)\del X^0_t +\sum_{i = 1}^k L_i(q_t)\del X^i_t \right.\nonumber\\ &\left.+ \dotp{p_t}{\del q_t - v_t\del X^0_t - \sum_{i = 1}^k V_i(q_t)\del X^i_t}\right).
    \end{align}
\end{definition}
Typically, we are interested in the case $X^0 = t$ and $X^i$ is a Brownian motion. 

\medskip

We now state the stochastic Hamilton-Pontryagin principle:
\begin{theorem}[Stochastic Hamilton-Pontryagin Principle]
    For every semimartingale $X = (X^0, \cdots, X^k)$ on $\R^{k+1}$, if $\Gamma_t = (q_t, v_t, p_t)\in \mathscr{S}(\P Q)$ is admissible then $\D \Ac{X}{\Gamma} = 0$ for all deformations $\epsilon\mapsto \Gamma_{\epsilon}$ such that $\delta q_t$ vanishes at $t = 0$ and $t= T$, if and only if, $\stopproc{\Gamma}{T} = \left(q_t^{|T}, v_t^{|T}, p_t^{|T}\right)$ satisfies the \textbf{stochastic implicit Euler-Lagrange equations} given by
    \begin{align}\label{eq:sellocal}
    \del q_t &= v_t\del X^0_t + \sum_{i=1}^kV_i(q_t)\del X^i_t\nonumber\\
    \del p_t &= \frac{\partial}{\partial q_t}\left(\L\del X_t^0 + \sum_{i = 1}^k\left(L_i - \dotp{p_t}{V_i(q_t)}\right)\del X^i_t\right)\nonumber\\
    \left(p_t - \pard{\L}{v_t}\right)\del X^0_t&=0.
\end{align}
\end{theorem}
Note that if $X^0 = t$ and $X^i = 0$ for all $i = 1, \cdots, k$ then the stochastic Hamilton-Pontryagin action and the implicit Euler-Lagrange equations agree with their deterministic counterparts respectively.

\section{Connectors and the Stochastic Covariant Derivative}
In the stochastic case, the covariant derivatives in the vertical and horizontal Lagrange-Poincaré equations are replaced by stochastic covariant derivatives. On a vector bundle with a connection, this is defined in terms of the connector or the connection map. 

\subsection{The Connector on a Vector Bundle}

Let $E\xrightarrow{\pi_E} M$ be a vector bundle with a connection $\nabla$ and let $GL(E)$ denote the frame bundle of $E$. From the connection $\nabla$, we obtain a projection $\Phi_E: TE\rightarrow \Ver E$ to the vertical bundle $\Ver E$ of $TE$. Let $\vlift^E: E\times_M E\rightarrow \Ver E$ denote the vertical lift map given by 
\begin{equation}
    \vlift_m^E(v_{m_1}, v_{m_2}) = \frac{d}{ds}\Big|_{s = 0}(v_{m_1} + sv_{m_2}),
\end{equation}
for all $m\in M$, $v_{m_1}, v_{m_2} \in \pi_E^{-1}(m)$. For proof of the next lemma, we refer the reader to Michor \cite{michor2008topics}.
\begin{lemma}\label{lem:propertiesofvlift}
    The vertical lift map $\vlift^E$ satisfies the following properties:
    \begin{enumerate}
        \item $\vlift^E$ is a vector bundle isomorphism.
        \item Let $\Pr_i$ denote the projection onto the $i$-th factor of $E\times_M E$. Then $\Pr_1\circ {(\vlift^{E})}^{-1} = \pr{TE}|_{\Ver E}$, where $\pr{TE}:TE\rightarrow E$ is the projection onto the basepoint. 
    \end{enumerate}
\end{lemma}
\begin{definition}
    The map $\vpr_E:= \Pr_2\circ {(\vlift^E)}^{-1}: \Ver E\rightarrow E$ is called the \textbf{vertical projection}. The map $K_{\nabla}:= \vpr_E\circ \Phi_E: TE\rightarrow E$ is called the \textbf{connector} associated to $\nabla$.
\end{definition}
Let $F:E\rightarrow \R$ and $v_0 \in E$ be an arbitrary point with $\pi_E(v_0) = m_0$. Let $u_{m_0} \in T_{m_0} M$. We define 
\begin{equation}
    \dotp{\frac{\partial F}{\partial m}\Big|_{v_{0}}}{(m_0, u_{m_0})} = \frac{d}{dt}\Big|_{t= 0} F(m^h(t))
\end{equation}
where $m^h(t)$ is the horizontal lift of a curve $m(t)$ in $M$, with $m(0)= m_0$, $\dot{m}(0) = u_{m_0}$ and $m^h(0) = v_0$. 
\begin{remark}
    This is independent of the choice of the curve $m(t)$ in $M$, since the tangent vector $\dot{m}^h(0)$ depends only on the horizontal lift of $u_{m_0}$.
\end{remark}
We also define 
\begin{equation}
    \dotp{\frac{\partial F}{\partial v}\Big|_{v_0}}{v} = \frac{d}{ds}\Big|_{s = 0} F(v_0 +sv) = \dotp{dF}{ \vlift^E_{m_0}(v_0, v)},
\end{equation}
for all $v\in\pi_E^{-1}(m_0)$. Note that the definition of $\frac{\partial F}{\partial v}$ is independent of $\nabla$.
\begin{theorem}\label{thm:variation_vector_bundle}
    Suppose $\epsilon\mapsto v_{m,\epsilon}(t)$ is a deformation of $v_m(t)$. Then
    \begin{equation}\label{eq:variation_vector_bundle}
        \dotp{dF}{\delta v_m(t)} = \dotp{\frac{\partial F}{\partial m}}{\delta m(t)} + \dotp{\frac{\partial F}{\partial v}}{K_{\nabla}(\delta v_{m}(t))}.
    \end{equation}
\end{theorem}
\begin{proof}
    We have
    \begin{equation*}
        \dotp{dF}{\delta v_m(t)} = \dotp{dF|_{\Hor E}}{\Hor\delta v_m(t)} + \dotp{dF|_{\Ver E}}{\Ver\delta v_m(t)}
    \end{equation*}
    where $\Hor\delta v_m(t)$ and $\Ver\delta v_m(t)$ are the horizontal and vertical components of $\delta v_m(t)$ respectively. For the first term on the right, horizontally lift $m_{\epsilon}(t) := \pi_E(v_{m,\epsilon}(t))$ to a curve $m^h_{\epsilon}(t)$ starting at $v_{m,\epsilon}(0)$. Then
    \begin{equation*}
        \dotp{dF|_{\Hor E}}{\Hor\delta v_m(t)} = \dotp{dF}{\delta m^h(t)} = \frac{d}{d\epsilon}\Big|_{\epsilon = 0}F(m^h_{\epsilon}(t)) = \dotp{\frac{\partial F}{\partial m}}{\delta m(t)}.
    \end{equation*}
    For the second term on the right, let $w(t) = {(\vlift^E)^{-1}}\circ \Phi_E(\delta v_m(t))$. Then, by Lemma \ref{lem:propertiesofvlift} $\Pr_1(w(t)) = v_m(t)$, and by definition of $K_{\nabla}$, $\Pr_2(w(t)) = K_{\nabla}(\delta v_m(t))$. Therefore, we have
    \begin{align*}
        \dotp{dF|_{\Ver E}}{\Ver\delta v_m(t)} &= \dotp{dF_{\Ver E}}{\Phi_E(\delta v_m(t))}\\
        &= \left\langle dF|_{\Ver E},(\vlift^E)(\mathrm{Pr}_1(w(t)), \mathrm{Pr}_2(w(t)))\right\rangle\\
        &= \dotp{dF|_{\Ver E}}{(\vlift^E)(v_m(t),K_{\nabla}(\delta v_m(t))}\\
        &= \dotp{dF|_{\Ver E}}{\frac{d}{ds}\Big|_{s = 0}(v_m(t) + K_{\nabla}(\delta v_m(t)))}\\
        &= \frac{d}{ds}\Big|_{s= 0} F(v_m(t) + K_{\nabla}(\delta v_m(t)))\\
        &= \dotp{\frac{\partial F}{\partial v}}{K_{\nabla}(\delta v_{m}(t))}.
    \end{align*}
    This completes the proof.
\end{proof}

\subsection{The Stochastic Covariant Derivative}\label{sec:stochastic_covariant_derivative}

Let $E\xrightarrow{\pi_E} M$ be a vector bundle with a connection $\nabla$ and let $E^*$ denote the dual bundle. Let $K_{\nabla}$ denote the associated connector. Suppose $\Gamma$ is a semimartingale in $E$. We shall take Proposition 1 in Catuogno, Ledesma, and Ruffino \cite{catuogno2013note} as the definition of the (Stratonovich) stochastic covariant derivative of $\Gamma$:
\begin{definition}
    For every section $\alpha: M\rightarrow E^*$ of the dual bundle $E^*$, we define the \textbf{stochastic covariant derivative} $\int \alpha\Del \Gamma$ by 
    \begin{equation}\label{eq:definition_stochastic_covariant_derivative}
\int \dotp{\alpha}{ \Del\Gamma} := \int \alpha(\pi_E(\Gamma))\circ K_{\nabla}\del \Gamma.\end{equation}
\end{definition}
We denote Equation \eqref{eq:definition_stochastic_covariant_derivative} in differential notation by 
\begin{equation*}
    \Del \Gamma: = K_{\nabla}\del \Gamma.
\end{equation*}
Either by purely deterministic arguments or as a special case of this definition, we obtain
\begin{equation}\label{eq:connectionmap}
    \frac{Dv(t)}{Dt}= K_{\nabla}(\dot{v}(t)).
\end{equation}

It follows from the expression of the covariant derivative on the associated bundle that given a semimartingale $[q_t, \bar{\zeta}_t]$ in $\alag$, we have
\begin{equation}\label{eq:stochastic_covariant_derivative_associated_bundle}
    \Del [q_t, \bar{\zeta}_t]_G = [q_t, [\zeta_t, A(q_t, \del q_t)] + \del \zeta_t]_G.
\end{equation}
Also note that if $K_{\nabla_{\alag}}$ denotes the connector on $\alag$ and we have a connection ${\nabla_{\dalag}}$ on the dual bundle $\dalag$ defined by Equation \eqref{eq:covariantderivativedualbundle} then 
\begin{equation}\label{eq:covariant_Derivative_dual_bundle_connector}
\frac{d}{dt}\dotp{\bar{\mu}(t)}{\bar{\zeta}(t)} = \dotp{\bar{\mu}(t)}{K_{\nabla^{\alag}}\dot{\bar{\zeta}}(t)} +  \dotp{K_{\nabla^{\dalag}}\dot{\bar{\mu}}(t)}{\bar{\zeta}(t)},
\end{equation}
where $\bar{\zeta}(t)$ and $\bar{\mu}(t)$ are curves in $\alag$ and $\dalag$ respectively that projects to the same curve in $Q/G$. For semimartingales $\bar{\zeta}\in\mathscr{S}(\alag)$ and $\bar{\mu}\in\mathscr{S}(\dalag)$ that have the same projection to $Q/G$, we define $\int \dotp{\Del \bar{\mu}}{\bar{\zeta}}$ by 
\begin{equation}\label{eq:stochastic_covariant_derivative_dual_bundle}
    \dotp{\bar{\mu}}{\bar{\zeta}} = \int \dotp{\Del \bar{\mu}}{\bar{\zeta}} + \int \dotp{\bar{\mu}}{\Del \bar{\zeta}}.
\end{equation}

\section{Stochastic Implicit Lagrange-Poincaré Reduction}

\subsection{The Reduced Action Functional}

We now assume that $\L\in\Cin(TQ)$, $L_i\in\Cin(Q)$ and $V_i\in\VecF{Q}$ are $G$-invariant. Let $l\in\Cin(T(Q/G)\oplus \alag)$ denote the reduced Lagrangian corresponding to $L$. By $G$-invariance of $L_i$, there exists $\ell_i\in \Cin(Q)$ such that $\ell_i\circ\pi = L_i$. Similarly, $G$-invariance of $V_i$ implies that there exists sections $\beta_i:Q/G\rightarrow \alag$ and vector fields $\Theta_i\in \VecF{Q/G}$ such that $\beta_i([q]) = [q, A(V_i(q))]_G$ and $\Theta_i \circ \pi = T\pi \circ V_i$. Note that $G$-invariance of $V_i$ and equivariance of $A$ implies that $\beta_i$ is well-defined. Similar to the deterministic reduced action functional $\S^{red}$ in Equation \eqref{eq:deterministic_reduced_action}, we shall describe a reduced stochastic action $\S_X^{red}$ corresponding to the stochastic Hamilton-Pontryagin action $\S_X$ in Equation \eqref{eq:stochastic_Hamilton-Pontryagin_action}.

\medskip

Denote by $G_{\P Q}: \P Q\rightarrow \R$ the pairing map $(u_q, p_q)\in \P Q\mapsto \dotp{p_q}{u_q}$. Let $G_{\P (Q/G)}$ denote the corresponding map for $\P (Q/G)$. Let $\rho_{T\cotQ}: T\cotQ\rightarrow \P Q$ denote the map $(q, p,u_q, u_p)\in T\cotQ\mapsto (q, u_q, p)\in \P Q$. Yoshimura and Marsden \cite{yoshimura061} show that this map is coordinate independent. Next, let \begin{align*}
    \pr{Q}&:\P Q\rightarrow Q\\
    \pr{TQ}&: \P Q\rightarrow TQ\\
    \pr{\cotQ}&: \P Q\rightarrow \cotQ
\end{align*}
denote the projections $(q, v, p)\in \P Q\mapsto q\in Q$, $(q, v, p)\in \P Q\mapsto (q,v)\in TQ$ and $(q, v, p)\in \P Q\mapsto (q, p)\in \cotQ$, respectively. The corresponding projections in $\P (Q/G)$ are denoted by $\pr{Q/G}$, $\pr{T(Q/G)}$, and $\pr{T^*(Q/G)}$ , respectively. Consider the 1-form $\mathcal{G}_{\P Q}$ on $\P Q$ defined by the composition $\mathcal{G}_{\P Q} = G_{\P Q}\circ \rho_{T\cotQ} \circ T\pr{\cotQ}$. The corresponding 1-form on $\P (Q/G)$ is denoted by $\mathcal{G}_{\P(Q/G)}$. 

\medskip

Given a vector field $V\in \VecF{Q}$, let $\bar{V}:\P Q\rightarrow \P Q$ denote the map defined by $\bar{V}(q,v,p) = (q,V(q), p) = (V\circ\pr{Q})\oplus \pr{\cotQ}$. If $V$ is a vector field on $Q/G$ instead, then we will slightly abuse notation to still denote by $\bar{V}:\P (Q/G)\rightarrow \P (Q/G)$ the map $\bar{V}(x, u, y) = (u, V(x), y)$, where $(x,u,y)\in \P (Q/G)$. We define generalized energies $E_j\in \Cin(\P Q)$ by 
\[E_j = \begin{cases}
    G - \L\circ \pr{TQ}, \:\text{if }j = 0,\\
    G\circ \bar{V}_j - L_j\circ\pr{Q},\:\text{if }j = 1,\cdots, k.
\end{cases}\]
In local coordinates, we have
\[E_j(q,v,p) = \begin{cases}
    \dotp{p_q}{v_q} - L(q,v_q), \:\text{if }j = 0,\\
    \dotp{p_q}{V_j(q)} - L_j(q),\:\text{if }j = 1,\cdots, k.
\end{cases}\]
In Saha \cite{saha2025stochastic} it is shown that the stochastic Hamilton-Pontryagin action functional can be written as
\begin{equation}\label{eq:action_functional_intrinsic}
    \Ac{X}{\Gamma} = \int_0^T \mathcal{G}_{\P Q}\del \Gamma - \sum_{j = 0}^k\int_0^TE_j(\Gamma)\del X^j.
\end{equation}
Now suppose $\L, L_1,\cdots, L_k$ and $V_1, \cdots, V_k$ are all $G$-invariant. Then, for all $g\in G$, $\Ac{X}{g\Gamma} = \Ac{X}{\Gamma}$. Therefore $\Ac{X}{\cdot}$ drops to an action $\S_{X}^{\reduced}({\cdot})$ on $\P Q/G$ defined by 
\[\S_X^{\reduced}([\Gamma]) = \Ac{X}{\Gamma}.\]
Via the isomorphism $\P Q/G\cong\P(Q/G)\oplus\alag\oplus\dalag$, we can identify $[\Gamma]$ with a semimartingale $\Gamma^{\P(Q/G)}\oplus \bar{\zeta}\oplus \bar{\mu}$, where $\Gamma^{\P(Q/G)}\in \mathscr{S}(\P(Q/G))$, $\bar{\zeta}\in\mathscr{S}(\alag)$ and $\bar{\mu}\in\mathscr{S}(\dalag)$. 

\medskip

Let $\ell\in\Cin(T(Q/G)\oplus\alag)$ denote the reduced Lagrangian corresponding to $\L$, $l_i\in\Cin(Q/G)$ denote the reduced Lagrangians corresponding to $L_i$, $V_i^{\reduced}\in\VecF{Q/G}$ denote the reduced vector fields corresponding to $V_i$, and $\bar{\beta}_i:Q/G\rightarrow \alag$ denote the section $\bar{\beta}_i([q]) = [q, \beta_i(q)]_G$, where $\beta_i(q) = A(V_i(q))$. Note that $\bar{\beta}_i$ is well-defined since $V_i$ is symmetric under the $G$-action and $A$ is $G$-equivariant. We define $\mathcal{E}_j\in\Cin(\P(Q/G)\oplus\alag\oplus\dalag)$ by 
\[\mathcal{E}_j(x, u , y, \bar{\zeta}, \bar{\mu}) = \begin{cases}
    \dotp{y}{u} + \dotp{\bar{\mu}}{\bar{\zeta}} - \ell(x, u, \bar{\zeta}),\:\text{if }j = 0,\\
    \dotp{y}{V_j^{\reduced}(x)} + \dotp{\bar{\mu}}{\bar{\beta}_j(x)} - l_j(x),\:\text{if }j = 1,\cdots,k.
\end{cases}\]
Denote by $G_{\alag}:\alag\oplus\dalag\rightarrow \R$ the fibrewise pairing map $(\bar{\zeta}_x, \bar{\mu}_x)\mapsto \dotp{\bar{\zeta}_x}{\bar{\mu}_x}$, where $\bar{\zeta}_x, \bar{\mu}_x$ are elements in the fibre over $x\in Q/G$ respectively. If $\bar{\mu}$ is a semimartingale in $\dalag$ and $q_t$ is a semimartingale in $Q$ with $\bar{\mu}_t$ projecting onto $\pi(q_t)$, we consider the semimartingale $\Psi_{\bar{\mu}}\in \cotQ$ over $q_t$ defined by 
\begin{equation}\label{eq:definition_of_Psi_bar_mu}
    \Psi_{\bar{\mu}_t}(v_{q_t}) = G_{\alag}(\bar{\mu}_t, [q_t, A(q_t, v_{q_t})]_G)
\end{equation}
for all $v_{q_t}\in T_{q_t} Q$. 
\begin{proposition}
    Let $q^1_t$ and $q^2_t$ be semimartingales in $Q$ such that $q_{1_t} = g_t q_{2_t}$ for some semimartingale $g_t$ in $G$. Denote by $\Psi^1_{\bar{\mu}_t}$ and $\Psi^2_{\bar{\mu}_t}$ the corresponding $T^*Q$-valued semimartingales over $q^1_t$ and $q^2_t$ respectively. Then \[\int \Psi^1_{\bar{\mu}_t}\del q^1_t = \int \Psi^2_{\bar{\mu}_t}\del q^2_t.\]
\end{proposition}
\begin{proof}
    Let $v_{q^1_t}$ be an arbitrary tangent vector in $T_{q^1_t} Q$. Then $v_{q^2_t} = g_tv_{q^1_t}$ is an arbitrary tangent vector in $T_{q^2_t}Q$. It suffices to show that $\Psi^1_{\bar{\mu}_t}(v_{q^1_t}) = \Psi^2_{\bar{\mu}_t}(v_{q^2_t})$. We have,
    \begin{align*}
        \Psi^2_{\bar{\mu}_t}(v_{q^2_t}) &= G_{\alag}(\bar{\mu}_t, [q^2_t, A(q^2_t,v_{q^2_t})]_G)\\
        &= \dotp{\bar{\mu}_t}{[g_tq^1_t, A(g_tq^1_t, g_t v_{q^1_t})]_G}\\
        &= \dotp{\bar{\mu}_t}{[g_tq^1_t, \Ad_{g_t}A(q^1_t, v_{q^1_t})]_G}\\
        &= \dotp{\bar{\mu}_t}{[q^1_t, A(q^1_t, v_{q^1_t})]_G}\\
        &= \Psi^1_{\bar{\mu}_t}(v_{q^1_t}),
    \end{align*}
    as required.
\end{proof}
Given a semimartingale $q_t$ in $Q$, we let $\bar{\xi}_t = [q_t,\xi_t]_G$, where $\xi_t =\int A\del q_t$. Let us define,
\begin{equation}
    \int \dotp{\bar{\mu}_t}{\del \bar{\xi}_t} = \int \Psi_{\bar{\mu}_t}\del q_t.
\end{equation}
\begin{remark}
    If $\Gamma$ is a semimartingale in a manifold $M$, we usually think of $\del \Gamma$ as a formal tangent vector in $TM$. However, in this case, we will think of $\del\bar{\xi}$ formally as an element of $\alag$, rather than $T\alag$. To motivate this interpretation, let us note that if $\xi_t = \int A(q_t)\del q_t$ then $\del \xi_t = A(q_t)\del q_t$ is a formal element in $\lag$. When $q_t = q(t)$ is a deterministic curve, we have $\dot{\xi}(t) = A(q(t), \dot{q}(t))$. The curve $A(q(t), \dot{q}(t))$ is referred to as $\xi(t)$ in Yoshimura and Marsden \cite{yoma}, and they set $\bar{\xi}(t) = [q(t), \xi(t)]_G$. In our case, we replace $\xi(t)$ by $\dot{\xi}(t)$ so that $\bar{\xi}(t) = [q(t), \dot{\xi}(t)]_G$. But when $q_t$ is a semimartingale, we replace $\dot{\xi}(t)=A(q(t), \dot{q}(t)) $ by $\del \xi_t = A(q_t)\del q_t$. Then, the right side is the formal Stratonovich differential $[q_t, A(q_t)\del q_t]_G$, and hence we must replace $\bar{\xi}(t)$ by a Stratonovich differential $\del \bar{\xi}_t$. The definition of $\int \dotp{\bar{\mu}}{\del \bar{\xi}_t}$ interprets the pairing of $\del \bar{\xi}_t$ with a semimartingale in $\dalag$ as a well-defined Stratonovich integral in $Q$.
\end{remark}
\begin{theorem}\label{thm:the_reduced_action_functional}
The reduced action functional $\S_X^{\reduced}(\cdot)$ is given by
\begin{align}\label{eq:the_reduced_action_functional}
    \S_X^{\reduced}(\Gamma^{\P(Q/G)}\oplus \bar{\zeta}\oplus \bar{\mu}) &= \int_0^T\mathcal{G}_{\P(Q/G)}\del \Gamma^{\P(Q/G)} + \int_0^T\dotp{\bar{\mu}}{\del\bar{\xi}}\nonumber\\& - \sum_{j = 0}^k\mathcal{E}_j(\Gamma^{\P(Q/G)}, \bar{\zeta},\bar{\mu})\del X^j.
\end{align}
\end{theorem}
\begin{proof}
    Let $(q, v, p)\in \P Q$, $x = \pi(q)$, $u = T_q\pi(v)$, $y = (p)_q^{h^*}$, $\bar{\zeta} = [q, A(q,v)]$ and $\bar{\mu} = [q, J(q, p)]_G$. Then, we have, \[\dotp{p}{v} = \dotp{p}{\Hor(v)} + \dotp{p}{\Ver(v)} = \dotp{y}{u} + \dotp{\bar{\mu}}{\bar{\zeta}},\]and similarly,
    \[\dotp{p}{V_i(q)} = \dotp{p}{\Hor V_i(q)} + \dotp{p}{\Ver V_i(q)} = \dotp{y}{V_i^{\reduced}(x)} + \dotp{\bar{\mu}}{\bar{\beta}_i(x)}.\]
    As a result,
    \begin{align*}E_j(q,v,p) &= \begin{cases}
    \dotp{p_q}{v_q} - L(q,v_q), \:\text{if }j = 0,\\
    \dotp{p_q}{V_j(q)} - L_j(q),\:\text{if }j = 1,\cdots, k.
\end{cases}\\
&= \begin{cases}
    \dotp{y}{u} + \dotp{\bar{\mu}}{\bar{\zeta}} - \ell(x, u, \bar{\zeta}),\:\text{if }j = 0,\\
    \dotp{y}{V_j^{\reduced}(x)} + \dotp{\bar{\mu}}{\bar{\beta}_j(x)} - l_j(x),\:\text{if }j = 1,\cdots,k.
\end{cases}\\
&= \mathcal{E}_j(x, u , y, \bar{\zeta}, \bar{\mu}).\end{align*}
Next suppose $(\dot{q}, \dot{v}, \dot{p})\in T_{(q,v,p)}\P Q$. Let $\pi_{\P(Q/G)}: \P (Q/G)\oplus \alag\oplus\dalag \rightarrow \P(Q/G)$ denote the projection onto the $\P(Q/G)$ factor. Then, we obtain a map $\Pi_{\P (Q/G)}: \P Q\rightarrow \P (Q/G)$ by $\Pi_{\P(Q/G)} = \pi_{\P(Q/G)}\circ\tilde{\Psi}_A\circ\Pr_{\P Q/G}$. Let $(\dot{x}, \dot{u}, \dot{y}) = T_{(q, v, p)}\Pi_{\P (Q/G)}(\dot{q}, \dot{v}, \dot{p})\in T_{(x, u, y)}\P(Q/G)$. Then
\begin{align*}
    \mathcal{G}_{\P Q}(q, v, p)(\dot{q}, \dot{v}, \dot{p}) &= \dotp{p}{\dot{q}}\\
    &= \dotp{p}{\Hor \dot{q}} + \dotp{p}{\Ver\dot{q}}\\
    &= \dotp{y}{\dot{x}} + \dotp{\bar{\mu}}{[q,A(q, \dot{q})]_G}\\
    &= \mathcal{G}_{\P(Q/G)}(x, u, y)(\dot{x}, \dot{u}, \dot{y}) + \Psi_{\bar{\mu}}(\dot{q}).
\end{align*}
Consequently, we obtain
\begin{align*}
    \S_X^{\reduced}(\Gamma^{\P(Q/G)}\oplus \bar{\zeta}\oplus \bar{\mu}) &= \int_0^T\mathcal{G}_{\P(Q/G)}\del \Gamma^{\P(Q/G)} + \int_0^T\dotp{\bar{\mu}}{\del\bar{\xi}}\nonumber\\& - \sum_{j = 0}^k\mathcal{E}_j(\Gamma^{\P(Q/G)}\oplus\bar{\zeta}\oplus\bar{\mu})\del X^j.
\end{align*}
    
\end{proof}
If $\Gamma_t = (q_t, v_t, p_t)$ in local coordinates and $\tilde{\Psi}_A \circ \Pr_{\P Q} (\Gamma_t) = (x_t, u_t, y_t, \bar{\zeta}_t, \bar{\mu}_t)$, then the local coordinate expression of $\S^{\reduced}_X(\cdot)$ is given by 
\begin{align}\label{eq:reduced_action_local_coordinates}
\S_X^{\reduced}(x_t, u_t, y_t, \bar{\zeta}_t, \bar{\mu}_t) = &\int_0^T \ell(x_t, u_t, \bar{\zeta}_t)\del X^0_t + \sum_{i = 1}^k l_i(x_t) \del X^i_t + \dotp{y_t}{ \del x_t - u_t \del X^0_t - \sum_{i = 1}^kV_i^{\reduced}(x_t)\del X^i_t} \nonumber\\&+ \dotp{ \bar{\mu}_t}{ \del \bar{\xi}_t - \bar{\zeta}_t\del X^0_t - \sum_{i=1}^k \bar{\beta}_i(x_t) \del X^i_t}.
\end{align}

\subsection{Vertical and Horizontal Variations}
Let us consider a semimartingale $q_t$ in $Q$. Define 
\[\xi_t = \int A(q_t)\del q_t \in \mathscr{S}(\lag).\]

\begin{theorem}\label{thm:horizontalandverticalvariations}
    Suppose $\epsilon\mapsto q_{t,\epsilon}$ is a deformation of $q_t$.
    \begin{enumerate}
        \item If $\delta q_t$ is vertical then
                \begin{equation}\label{eq:stochastic_vertical_variations}
                    \D \xi_t = \int \ad_{\eta_t}\del\xi_t + \eta_t - \eta_0
                \end{equation}
                where $\eta_t = A(q_t, \delta q_t)$.
        \item If $\delta q$ is horizontal then
        \begin{equation}\label{eq:stochastic_horizontal_variations}
            \D \xi_t = \int B(q_t)(\delta q_t, \del q_t):= \int i_{\delta q_t}B\del q_t.
        \end{equation}
        \end{enumerate}
\end{theorem}
\begin{proof}
    \begin{enumerate}
        \item Since $\Hor \delta q_t = 0$ it follows from \eqref{eq:curvatureformulas} that $B(q_t)(\delta q_t, \cdot) = 0$. Hence, by Lemma \ref{lem:variationsofsemimartginales}
\begin{align*}
   \D \xi_t &= \int i_{\delta q_t}dA \del q_t + \dotp{A(q_t)}{\delta q_t} - \dotp{A(q_0)}{\delta q_0}\\
   &= \int \ad_{\eta_t}A(q_t)\del q_t+ \eta_t - \eta_0\\
   &= \int \ad_{\eta_t}\del \int A(q_t)\del q_t+ \eta_t - \eta_0\\
   &= \int \ad_{\eta_t}\del \xi_t + \eta_t - \eta_0.
\end{align*}
\item Since $A(q_t,\delta q_t) = 0$, we have
\begin{align*}
    \D \xi_t &= \int i_{\delta q_t}dA \del q_t + \dotp{A(q_t)}{\delta q_t} - \dotp{A(q_0)}{\delta q_0}\\
    &= \int i_{\delta q_t} B \del q_t\\
    &= \int B(q_t)(\delta q_t, \del q_t).
\end{align*}
    \end{enumerate}
\end{proof}

\subsection{Covariant Variations on the Adjoint Bundle}\label{sec:covariant_variations_adjoint_bundle}

In this section, we discuss variation of the term $\int \dotp{\bar{\mu}_t}{\del\bar{\xi}_t}$ in the reduced action functional.
\begin{theorem}
Let $q_t$ be a semimartingale in $Q$, $\bar{\mu}_t$ be a semimartingale in $\dalag$ that projects to $x_t := \pi(q_t)$ in $Q/G$ and $\bar{\xi}_t = [q_t, \xi_t]_G$ with $\xi_t = \int A(q_t) \del q_t$. Suppose $\epsilon\mapsto q_{\epsilon,t}$ and $\epsilon\mapsto \bar{\mu}_{\epsilon,t}$ are deformations of $q_t$ and $\bar{\mu}_t$ respectively. Set $x_{\epsilon,t} = \pi(q_{\epsilon,t})$, $\eta_t = A(q_t, \delta q_t)$ and $\bar{\eta}_t = [q_t, \eta_t]_G$. Then
\begin{align}\label{eq:covariantvariationontheadjointbundle}
    \D\int \dotp{\bar{\mu}_t}{\del \bar{\xi}_t} &=
 \int \dotp{K_{\nabla^{\dalag}} \delta \bar{\mu}_t - \ad^*_{\bar{\eta}_t}\bar{\mu}_t}{\del \bar{\xi}}+\int \dotp{\bar{\mu}_t}{i_{\delta x_t}\tilde{B}(x_t)}\del x_t- \int \dotp{\Del \bar{\mu}_t} {\bar{\eta}_t}\nonumber\\
    &+ \dotp{\bar{\mu}_t}{\bar{\eta}_t} - \dotp{\bar{\mu}_0}{\bar{\eta}_0}.
\end{align}
\end{theorem}
\begin{proof}
To make our calculations easier, we begin by considering the 1-form $\bar{\Psi}$ on $\dalag\times_{Q/G} Q$ defined by $\bar{\Psi}_{(\bar{\mu}, q)}{(v_{\bar{\mu}}, v_q)} = \Psi_{\bar{\mu}}(v_q) = \dotp{\bar{\mu}}{[q,A(q, v_q)]_G}$, where $(\bar{\mu}, q)\in  \dalag\times_{Q/G} Q$ and $(v_{\bar{\mu}}, v_q) \in T_{(\bar{\mu}, q)}(\dalag\times_{Q/G} Q)$. Then, we have,
\begin{equation}
    \int \dotp{\bar{\mu_t}}{\del \bar{\xi}_t} = \int \bar{\Psi} \del (\bar{\mu}_t, q_t).
\end{equation}
This implies, by Lemma \ref{lem:variationsofsemimartginales}, that
\[\D \int \dotp{\bar{\mu_t}}{\del \bar{\xi}_t} =  \int i_{(\delta \bar{\mu}_t, \delta q_t)}\d\bar{\Psi} \del (\bar{\mu}_t, q_t) + \dotp{\bar{\Psi}(\bar{\mu}_t, q_t)}{(\delta \bar{\mu}_t, \delta q_t)} - \dotp{\bar{\Psi}(\bar{\mu}_0, q_0)}{(\delta \bar{\mu}_0, \delta q_0)}.\]
We now compute $i_{(\delta \bar{\mu}_t, \delta q_t)}\d\bar{\Psi}$ by Corollary \ref{cor:importantcorollary}. Consider tangent vectors $(v_{\bar{\mu}}, v_q), (w_{\bar{\mu}}, w_q) \in T_{(\bar{\mu}, q)}(\dalag\times_{Q/G} Q)$.  Let $q_{\epsilon}(t)$ and $\bar{\mu}_{\epsilon}(t)$ be curves in $Q$ and $\dalag$ respectively that project on the same curve $x_{\epsilon}(t)$ and set $q_0(t) = q(t)$, $\bar{\mu}_0(t) = \bar{\mu}(t)$ and $x_0(t) = x(t)$. We assume that $q(0) = q$, $\bar{\mu}(0) = \bar{\mu}$, $\dot{q}(0) = v_q$, $\delta q(0) = w_q$, $\dot{\bar{\mu}}(0) = v_{\bar{\mu}}$ and $\delta \bar{\mu}(0) = w_{\bar{\mu}}$. Further let $x(0) = x$, $v_x = \dot{x}(0)$ and $w_x = \delta x(0)$. We have
\begin{align*}
    i_{(\delta \bar{\mu}(t), \delta q(t))}\d\bar{\Psi}(\dot{\bar{\mu}}, \dot{q}) &= \delta \dotp{\bar{\Psi}(\bar{\mu}(t), q(t))}{(\dot{\bar{\mu}}(t), \dot{q}(t))} - \frac{d}{dt}\dotp{\bar{\Psi}(\bar{\mu}(t), q(t))}{(\delta{\bar{\mu}}(t), \delta{q}(t))}\\
    &= \delta \dotp{\bar{\mu}(t)}{[q(t), A(q(t), \dot{q}(t))]_G} - \frac{d}{dt} \dotp{\bar{\mu}(t)}{[q(t), A(q(t), \delta q(t))]_G}.
\end{align*}
We now use Equation \eqref{eq:covariantderivativedualbundle} to obtain
\begin{align*}
    i_{(\delta \bar{\mu}(t), \delta q(t))}\d\bar{\Psi}(\dot{\bar{\mu}}, \dot{q}) &= \delta \dotp{\bar{\mu}(t)}{[q(t), A(q(t), \dot{q}(t))]_G} - \frac{d}{dt} \dotp{\bar{\mu}(t)}{[q(t), A(q(t), \delta q(t))]_G}\\
    &= \dotp{\Deldeleps \bar{\mu}_{\epsilon}(t)}{[q(t), A(q(t), \dot{q}(t))]_G} +\dotp{\bar{\mu}}{\Deldeleps[q(t), A(q(t), \dot{q}(t))]_G}\\
    &- \dotp{\frac{D}{Dt} \bar{\mu}(t)}{[q(t), A(q(t), \delta{q}(t))]_G} +\dotp{\bar{\mu}}{\frac{D}{Dt}[q(t), A(q(t), \delta{q}(t))]_G}\\
    &= \dotp{\Deldeleps \bar{\mu}_{\epsilon}(t)}{[q(t), A(q(t), \dot{q}(t))]_G} - \dotp{\bar{\mu}(t)}{[q(t), \ad_{A(q(t), \delta q(t))}A(q(t), \dot{q}(t)) ]_G} \\
    & + \dotp{\bar{\mu}}{\frac{D}{Dt}[q(t), A(q(t), \delta{q}(t))]_G} + \dotp{\bar{\mu}(t)}{[q(t), B(q(t))(\delta q(t),\dot{q}(t))]_G}\\
    &- \dotp{\frac{D}{Dt} \bar{\mu}(t)}{[q(t), A(q(t), \delta{q}(t))]_G}-\dotp{\bar{\mu}}{\frac{D}{Dt}[q(t), A(q(t), \delta{q}(t))]_G}.
\end{align*}
In the last step, we have used Equation \eqref{eq:total_covariant_variation} to write 
\begin{align*}\Deldeleps [q_{\epsilon}(t), A(q_{\epsilon}(t), \dot{q}_{\epsilon}(t))]_G &= \frac{D}{Dt}[q(t), A(q(t), \delta q(t)]_G - [q(t), \ad_{A(q(t), \delta q(t))}A(q(t), \dot{q}(t))]_G \\&+ [q(t), B(q(t))(\delta q(t), \dot{q}(t))]_G.\end{align*}
We observe that the term $\dotp{\bar{\mu}}{\frac{D}{Dt}[q(t), A(q(t), \delta{q}(t))]_G}$ cancels out. The covariant derivatives can be expressed by the connector via Equation \eqref{eq:connectionmap} to yield
\begin{align*}
    i_{(\delta \bar{\mu}(t), \delta q(t))}\d\bar{\Psi}(\dot{\bar{\mu}}, \dot{q})&= \dotp{\Deldeleps \bar{\mu}_{\epsilon}(t)}{[q(t), A(q(t), \dot{q}(t))]_G} - \dotp{\bar{\mu}(t)}{[q(t), \ad_{A(q(t), \delta q(t))}A(q(t), \dot{q}(t)) ]_G}\\
    &+ \dotp{\bar{\mu}(t)}{[q(t), B(q(t))(\delta q(t),\dot{q}(t))]_G} - \dotp{\frac{D}{Dt} \bar{\mu}(t)}{[q(t), A(q(t), \delta{q}(t))]_G}\\
    &= \dotp{K_{\nabla^{\dalag}} \bar{\mu}_{\epsilon}(t)}{[q(t), A(q(t), \dot{q}(t))]_G} - \dotp{\ad^*_{[q(t), A(q(t), \delta q(t))]_G}\bar{\mu}(t)}{A(q(t), \dot{q}(t)) ]_G}\\
    &+ \dotp{\dotp{\bar{\mu}(t)}{i_{\delta x(t)}\tilde{B}(x(t))}}{\dot{x}(t)} - \dotp{K_{\nabla^{\dalag}} \bar{\mu}(t)}{[q(t), A(q(t), \delta{q}(t))]_G}
\end{align*}
Evaluating at $t = 0$ we obtain
\begin{align*}
    i_{(w_{\bar{\mu}}, w_q)}\d\bar{\Psi}(v_{\bar{\mu}}, v_q)
    &= \dotp{K_{\nabla^{\dalag}} w_{\bar{\mu}}}{[q, A(q, v_q)]_G} - \dotp{\ad^*_{[q, A(q, w_q)]_G}\bar{\mu}}{A(q, v_q) ]_G}\\
    &+ \dotp{\dotp{\bar{\mu}}{i_{w_x}\tilde{B}(x)}}{v_x} - \dotp{K_{\nabla^{\dalag}} v_{\bar{\mu}}}{[q, A(q, w_q)]_G}.
\end{align*}
Since all of the tangent vectors chosen are arbitrary, we obtain
\begin{align*}
    \int i_{(\delta \bar{\mu}_t, \delta q_t)}\d \bar{\Psi} \del (\bar{\mu}_t, q_t) &= \int \Psi_{K_{\nabla^{\dalag}} \delta \bar{\mu}_t - \ad^*_{[q, A(q, \delta q_t)]_G}\bar{\mu}_t}\del q_t + \int \dotp{\bar{\mu}_t}{i_{\delta x_t}\tilde{B}(x_t)}\del x_t - \int \dotp{\Del \bar{\mu}_t} {[q_t, A(q_t, \delta q_t)]_G}\\
    &= \int \dotp{K_{\nabla^{\dalag}} \delta \bar{\mu}_t - \ad^*_{\bar{\eta}_t}\bar{\mu}_t}{\del \bar{\xi}}+\int \dotp{\bar{\mu}_t}{i_{\delta x_t}\tilde{B}(x_t)}\del x_t- \int \dotp{\Del \bar{\mu}_t} {\bar{\eta}_t},
\end{align*}
which yields
\begin{align*}
    \D\int \dotp{\bar{\mu}_t}{\del \bar{\xi}_t} &= \int \dotp{K_{\nabla^{\dalag}} \delta \bar{\mu}_t - \ad^*_{\bar{\eta}_t}\bar{\mu}_t}{\del \bar{\xi}}+\int \dotp{\bar{\mu}_t}{i_{\delta x_t}\tilde{B}(x_t)}\del x_t- \int \dotp{\Del \bar{\mu}_t} {\bar{\eta}_t}\\
    &+ \dotp{\bar{\Psi}(\bar{\mu}_t, q_t)}{(\delta \bar{\mu}_t, \delta q_t)} - \dotp{\bar{\Psi}(\bar{\mu}_0, q_0)}{(\delta \bar{\mu}_0, \delta q_0)}\\
    &= \int \dotp{K_{\nabla^{\dalag}} \delta \bar{\mu}_t - \ad^*_{\bar{\eta}_t}\bar{\mu}_t}{\del \bar{\xi}}+\int \dotp{\bar{\mu}_t}{i_{\delta x_t}\tilde{B}(x_t)}\del x_t- \int \dotp{\Del \bar{\mu}_t} {\bar{\eta}_t}\\
    &+ \dotp{\bar{\mu}_t}{\bar{\eta}_t} - \dotp{\bar{\mu}_0}{\bar{\eta}_0}.
\end{align*}
This completes the proof.
\end{proof}
Let us now show that the above formula can also be obtained by a formal computation by considering covariant variations of the semimartingale $\bar{\xi}_t = [q_t, \xi_t]_G$. Since $\epsilon\mapsto \bar{\xi}_{\epsilon,t}$ is pathwise smooth, the covariant derivative $\frac{D}{D\epsilon}\Big|_{\epsilon = 0}\bar{\xi}_{\epsilon,t}$ with respect to $\epsilon$ in the ucp topology agrees with the pathwise computed covariant $\frac{D}{D\epsilon}\Big|_{\epsilon = 0} \bar{\xi}_{t}(\cdot)$. This is a consequence of Arnaudon and Thalmaier \cite[Proposition 2.8]{arnaudon1998}. Using the covariant derivative on the adjoint bundle and Proposition \ref{thm:horizontalandverticalvariations}, we obtain
\begin{align*}
    \delta^A\bar{\xi}_t&\coloneqq \frac{D}{D\epsilon}\Big|_{\epsilon = 0}\bar{\xi}_{\epsilon,t}\\
    &= \left[q_t, -\ad_{\eta_t}\int A(q_t)\del q_t + \D \xi_t\right]_G\\
    &= \left[q_t, -\ad_{\eta_t}\int A(q_t)\del q_t\right]_G + \left[q_t, \int \del \eta_t + \int\ad_{\eta_t}\del \xi_t\right]_G\\
    &+ \left[q_t, \int i_{\delta q_t} B\del q_t\right]_G.
\end{align*}
We now formally set
\begin{align*}\int \dotp{\bar{\mu}_t}{\del (\delta^A\bar{\xi}_t)} &= \int \dotp{\bar{\mu}_t}{\left[q_t, -\ad_{\eta_t} A(q_t)\del q_t\right]_G} + \int \dotp{\bar{\mu}_t}{\left[q_t,  \del \eta_t + \ad_{\eta_t}\del \xi_t\right]_G}\\
&+\int \dotp{\bar{\mu}_t}{ \left[q_t,  i_{\delta q_t} B\del q_t\right]_G}.\end{align*}
To obtain each of the integrals on the right as Stratonovich integrals, we formally treat the Stratonovich differentials as tangent vectors in each of the pairings involved. For the first term, we have 
\begin{align*}
    \dotp{\bar{\mu}_t}{\left[q_t, -\ad_{\eta_t} A(q_t)\del q_t\right]_G}  &= -\dotp{\bar{\mu}_t}{\ad_{\bar{\eta}_t}[q_t, A(q_t)\del q_t]_G}\\
    &= -\dotp{\ad^*_{\bar{\eta}_t}\bar{\mu}_t}{[q_t, A(q_t)\del q_t]_G}
\end{align*}
which shows that $\int \dotp{\bar{\mu}_t}{\left[q_t, -\ad_{\eta_t} A(q_t)\del q_t\right]_G}$ corresponds to $-\int\dotp{\ad^*_{\bar{\eta}_t}\bar{\mu}_t}{\del \bar{\xi}_t} $. For the next term, we note that the stochastic covariant derivative of $\bar{\eta}_t$ on the adjoint bundle is given by 
\[\Del \bar{\eta}_t = [q_t, [{\eta}_t,A(q_t)\del q_t] + \del \eta_t]_G.\]
Hence, $\int \dotp{\bar{\mu}_t}{\left[q_t,  \del \eta_t + \ad_{\eta_t}\del \xi_t\right]_G}$ corresponds to $\int \dotp{\bar{\mu}}{\Del \bar{\eta}}$. For the final term on the right side, we have
\begin{align*}
    \dotp{\bar{\mu}_t}{ \left[q_t,  i_{\delta q_t} B\del q_t\right]} &= \dotp{\bar{\mu}_t}{\tilde{B}(x_t)(\delta x_t, \del x_t)}\\
    &= \dotp{\dotp{\bar{\mu}_t}{i_{\delta x_t}\tilde{B}(x_t)}}{\del x_t}
\end{align*}
and hence, we can write
\[\int\dotp{\bar{\mu}_t}{ \left[q_t,  i_{\delta q_t} B\del q_t\right]} = \int \dotp{\bar{\mu}_t}{i_{\delta x_t}\tilde{B}(x_t)}\del x_t. \]
Thus, we obtain
\begin{align*}
    \int \dotp{\bar{\mu}_t}{\del (\delta^A\bar{\xi}_t)} &= -\int\dotp{\ad^*_{\bar{\eta}_t}\bar{\mu}_t}{\del \bar{\xi}_t} + \int \dotp{\bar{\mu}_t}{\Del \bar{\eta}_t} + \int \dotp{\bar{\mu}_t}{i_{\delta x_t}\tilde{B}(x_t)}\del x_t.
\end{align*}
which we will write as 
\begin{align*}
    {\del (\delta^A\bar{\xi}_t)} &= -\ad_{\bar{\eta}_t}{\del \bar{\xi}_t} + {\Del \bar{\eta}_t} + {i_{\delta x_t}\tilde{B}(x_t)}\del x_t.
    \end{align*}
to emphasize the similarity with the deterministic case (see Equation \eqref{eq:total_covariant_variation}). By Equation \eqref{eq:stochastic_covariant_derivative_dual_bundle}, we have
\[\dotp{\bar{\mu}}{\bar{\eta}} - \dotp{\bar{\mu}_0}{{\bar{\eta}_0}} = \int\dotp{D\bar{\mu}}{\bar{\eta}} +\int \dotp{\bar{\mu}}{\Del \bar{\eta}}.\]
Putting all of this together, we have
\begin{align}\label{eq:auxiliary_eq_covariant_variation}
    \int \dotp{\bar{\mu}_t}{\del (\delta^A\bar{\xi}_t)} &= -\int\dotp{\ad^*_{\bar{\eta}_t}\bar{\mu}_t}{\del \bar{\xi}_t} - \int\dotp{D\bar{\mu}}{\bar{\eta}} + \int \dotp{\bar{\mu}_t}{i_{\delta x_t}\tilde{B}(x_t)}\del x_t\nonumber\\
    &+ \dotp{\bar{\mu}_t}{\bar{\eta}_t} - \dotp{\bar{\mu}_0}{{\bar{\eta}_0}}.
\end{align}
Note that the terms on the right are well-defined Stratonovich integrals. By formally considering $\Deldeleps(\del \bar{\xi}_{\epsilon,t}) = \del(\delta^A\bar{\xi}_t)$, we obtain
\[\D \int \dotp{\bar{\mu}_t}{\del\bar{\xi}_t} = \int \dotp{K_{\nabla^{\dalag}}\delta\bar{\mu}_t}{\del\bar{\xi}_t} + \int \dotp{\bar{\mu}_t}{\del (\delta^A\bar{\xi}_t)}.\]
Replacing $\int \dotp{\bar{\mu}_t}{\del (\delta^A\bar{\xi}_t)}$ by the expression in Equation \eqref{eq:auxiliary_eq_covariant_variation} gives us the expression for $\delta\int\dotp{\bar{\mu}_t}{\del \bar{\xi}_t}$ obtained in Equation \eqref{eq:covariantvariationontheadjointbundle}. 
\begin{remark}\label{rem:stochastic_covariant_variations}
    As in the deterministic case, for the semimartingale $\bar{\xi}_t = [q_t, \int A(q_t)\del q_t]_G$, we will only consider deformations $\epsilon\mapsto \bar{\xi}_{\epsilon,t}$ such that their projection on $Q/G$ is $x_t = \pi(q_t)$. The corresponding variations can be identified with a semimartingale in $\alag$, rather than one in $T\alag$. The variation $\delta ^A\bar{\xi}$ is an instance of $\alag$-fiber variation. 

    \medskip
    
    We have already seen that if $q_t$ is a semimartingale such that $x_t = \pi(q_t)$ and $\bar{\xi}_t = [q_t, \int A(q_t)\del q_t]_G$ then, given a deformation $\epsilon\mapsto q_{\epsilon, t}$, we (formally) have
    \begin{align*}
    {\del (\delta^A\bar{\xi}_t)} &= -\ad_{\bar{\eta}_t}{\del \bar{\xi}_t} + {\Del \bar{\eta}_t} + {i_{\delta x_t}\tilde{B}(x_t)}\del x_t.
    \end{align*}
    Moreover, we also have
    \begin{align*}
        \D \int \dotp{\bar{\mu}_t}{\del\bar{\xi}_t} - \int \dotp{K_{\nabla^{\dalag}}\delta\bar{\mu}_t}{\del\bar{\xi}_t}
        &= \int \dotp{\bar{\mu}_t}{\del (\delta^A\bar{\xi}_t)}\end{align*}
    where 
    \begin{align*}
        \int \dotp{\bar{\mu}_t}{\del (\delta^A\bar{\xi}_t)}&=-\int \dotp{\ad^*_{\bar{\eta}_t}\bar{\mu}_t}{\del \bar{\xi}}+\int \dotp{\bar{\mu}_t}{i_{\delta x_t}\tilde{B}(x_t)}\del x_t- \int \dotp{\Del \bar{\mu}_t} {\bar{\eta}_t}\\
    &+ \dotp{\bar{\mu}_t}{\bar{\eta}_t} - \dotp{\bar{\mu}_0}{\bar{\eta}_0}.
    \end{align*}
    Thus, in the stochastic case, the variations $\delta^A\bar{\xi}_t$ are those $\alag$-fiber variations for which 
    \begin{align*}
        \D\int \dotp{\bar{\mu}_t}{\del\bar{\xi}_t} - \int \dotp{K_{\nabla^{\dalag}}\delta \bar{\mu}_t}{\del\bar{\xi}_t}
        &= \int \dotp{\bar{\mu}_t}{\Del\bar{\eta}_t} + \int \dotp{\bar{\mu}_t}{i_{\delta x_t}\tilde{B}(x_t)}\del x_t  -\int \dotp{\ad^*_{\bar{\eta}_t}\bar{\mu}_t}{\del \bar{\xi}}\\
        &= -\int \dotp{\ad^*_{\bar{\eta}_t}\bar{\mu}_t}{\del \bar{\xi}}+\int \dotp{\bar{\mu}_t}{i_{\delta x_t}\tilde{B}(x_t)}\del x_t- \int \dotp{\Del \bar{\mu}_t} {\bar{\eta}_t}\\
    &+ \dotp{\bar{\mu}_t}{\bar{\eta}_t} - \dotp{\bar{\mu}_0}{\bar{\eta}_0},
    \end{align*}
    where $\bar{\eta}_t = [q_t, A(q_t, \delta q_t)]_G$. The variations of the form $\delta x_t \oplus \delta^A\bar{\xi}_t$ are the stochastic analogues of the deterministic covariant variations. The reduced variational principle $\S_{X}^{\mathrm{red}}$ will be consider under constrained variations of this form. 
\end{remark}



\subsection{The Stochastic Implicit Lagrange-Poincaré Reduction Theorem}

In this section we will address the critical points of $S_X^{\mathrm{red}}$ given by Equation \eqref{eq:reduced_action_local_coordinates} and prove the stochastic version of the implicit Lagrange-Poincaré reduction theorem. 

\medskip

Fix a connection $\nabla^{T(Q/G)}$ on $T(Q/G)$, and let $T^*(Q/G)$ denote the connection on $T^*(Q/G)$ obtained from $T(Q/G)$. Then $ \grad^{T(Q/G)}\oplus\grad^{\alag}$ defines a connection on $\grad^{T(Q/G)}\oplus\grad^{\alag}$ with associated connector $K_{\nabla^{T(Q/G)}}\oplus K_{\nabla^{\alag}}$, where $K_{\nabla^{T(Q/G)}}$ and $K_{\nabla^{\alag}}$ are the connectors corresponding to $\grad^{T(Q/G)}$ and $\grad^{\alag}$ respectively. 

\begin{thm}\label{thm:variation_of_Sred}
Let $\Gamma_t = (x_t, u_t, y_t, \bar{\zeta}_t, \bar{\mu}_t)$ be an admissible semimartingale in $\P (Q/G)\oplus \alag\oplus \dalag$, $K$ be a regular coordinate ball in $\P (Q/G)\oplus \alag\oplus \dalag$, $\tau_K^h$ be the hitting time of $\Gamma$ for $K$ and $\tau_{K}^{(h,e)}$ be the exit time of $\Gamma_{t + \tau_K^k}$ for $K$. Suppose $\epsilon \mapsto\Gamma^{|T}_{\epsilon,t} =(x^{|T}_{\epsilon,t}, u^{|T}_{\epsilon,t},y^{|T}_{\epsilon,t},  \bar{\zeta}^{|T}_{\epsilon,t}, \bar{\mu}^{|T}_{\epsilon,t})$ is a deformation of $\Gamma_t^{|T} = (x^{|T}_t, u^{|T}_t, y^{|T}_t, \bar{\zeta}^{|T}_t, \bar{\mu}^{|T}_t)$ such that the variations $\delta u^{|T}_t, \delta \bar{\zeta}^{|T}_t, \delta  y^{|T}_t$ and $\delta \bar{\mu}^{|T}_t$ are arbitrary and the variations $\delta x^{|T}_t\oplus \delta^A( \bar{\xi}^{|T}_t)$ satisfy $\delta x^{|T} = 0$ at $t = 0$ and $t = T$ and \begin{align*}
    {\del (\delta^A\bar{\xi}^{|T}_t)} &= -\ad_{\bar{\eta}_t}{\del \bar{\xi}^{|T}_t} + {\Del \bar{\eta}_t} + {i_{\delta x^{|T}_t}\tilde{B}(x_t)}\del x_t
    \end{align*} in the sense of Remark \ref{rem:stochastic_covariant_variations}. Here $\bar{\eta}$ is an arbitrary semimartingale in $\alag$ that vanishes at $t = 0$ and $t = T$. Then
\begingroup
\allowdisplaybreaks
\begin{align}\label{eq:variationinSred}
    \D \S^{\mathrm{red}}_X(x_{t}, u_{t}, y_{t}, \bar{\zeta}_t,\bar{\mu}_{t})&=
    \int_0^T \left\langle \frac{\partial}{\partial x_t^{|T}} \left( \ell \del X^0_t + \sum_{i =1}^k\left(l_i -  \dotp{y_t^{|T}}{V_i^{\reduced}(x_t^{|T})} - \dotp{\bar{\mu}_t^{|T}}{\bar{\beta}_i(x_t^{|T})}\right)\del X^i_t\right)\right.\nonumber\\
    &\left. - \dotp{\bar{\mu}_t}{i_{\del x_t}\tilde{B}(x_t)} - \Del y_t^{|T},\delta x_t^{|T}\right\rangle +\int_0^T\left\langle\frac{\partial \ell}{\partial u^{|T}_t} - y^{|T}_t, K_{\grad^{T(Q/G)}}\delta u^{|T}_t\right\rangle \del X^0_t\nonumber\\
    &+ \int_0^T\left\langle\frac{\partial \ell}{\partial {\bar{\zeta}}^{|T}_t} - \bar{\mu}^{|T}_t, K_{\grad^{\alag}}\delta {\bar{\zeta}}_t\right\rangle \del X^0_t\nonumber\\&+\int_0^T \langle K_{\grad^{T^*(Q/G)}} \delta y^{|T}_t, \del x^{|T}_t - u^{|T}_t\del X^0_t - \sum_{i = 1}^k V_i^{\mathrm{red}}(x^{|T}_t)\del X^i_t\rangle\nonumber\\ &+ \int_0^T\langle K_{\grad^{\dalag}}\delta \bar{\mu}^{|T}_t, \del \bar{\xi}^{|T}_t - \bar{\zeta}^{|T}_t\del X^0_t - \sum_{i = 1}^k\bar{\beta}_i(x^{|T}_t)\del X^i_t\rangle\nonumber\\&+\int_0^T \langle -\Del\bar{\mu}^{|T}_t+\ad^*_{\del \bar{\xi}^{|T}_t}\bar{\mu}^{|T}_t, \bar{\eta}_t\rangle
 \end{align}
 \endgroup
Consequently, $ \D \S^{\mathrm{red}}_X(x_{t}, u_{t},y_{t}, \bar{\zeta}_{t},  \bar{\mu}_{t}) = 0$ holds if and only if the \textbf{horizontal stochastic Lagrange-Poincaré equations} 
 \begin{align*}
     \Del y_t &= \frac{\partial}{\partial x_t} \left( \ell \del X^0_t + \sum_{i =1}^k\left(l_i -  \dotp{y_t}{V_i^{\reduced}(x_t)} - \dotp{\bar{\mu}_t}{\bar{\beta}_i(x_t)}\right)\del X^i_t\right) - \dotp{\bar{\mu}_t}{i_{\del x_t}\tilde{B}(x_t)} \\
     \del x_t &= u_t\del X^0_t + \sum_{i = 1}^kV_i^{\mathrm{red}}(x_t)\del X^i_t\\
     \left(y_t - \frac{\partial \ell}{\partial u_t}\right)\del X^0_t &= 0
 \end{align*}
and the \textbf{vertical stochastic implicit Lagrange-Poincaré equations}
\begin{align*}
    \Del \bar{\mu}_t &= \ad^*_{\del \bar{\xi}_t} \bar{\mu}_t\\
    \del \bar{\xi}_t &= \bar{\zeta}_t\del X^0_t + \sum_{i = 1}^k\bar{\beta}_i(x_t)\del X^i_t\\
    \left(\bar{\mu}_t - \frac{\partial \ell}{\partial \bar{\zeta}_t}\right)\del X^0_t &= 0.
\end{align*}
are satisfied by $(x_t^{|T}, u_t^{|T}, y_t^{|T}, \bar{\zeta}_t^{|T}, \bar{\mu}_t^{|T})$.
\end{thm}
\begin{proof}
    We refer the reader to Appendix A for proof of the theorem.
\end{proof}

We now summarize what we have obtained so far in the next theorem. This serves as the stochastic analogue of the deterministic implicit Lagrange-Poincaré reduction theorem (Theorem \ref{thm:deterministic_LP_reduction}):
\begin{theorem}[Stochastic Implicit Lagrange-Poincaré Reduction Theorem]
    The following are equivalent:
    \begin{enumerate}
        \item An admissible semimartingale $\Gamma_t = (q_t, v_t, p_t)$ in $\P Q$ is a critical point of the stochastic action integral
        \begin{align}
        \Ac{X}{\Gamma} = &\int_0^T\left(\L(q_t, v_t)\del X^0_t +\sum_{i = 1}^k L_i(q_t)\del X^i_t \right.\nonumber\\ &\left.+ \dotp{p_t}{\del q_t - v_t\del X^0_t - \sum_{i = 1}^k V_i(q_t)\del X^i_t}\right).
    \end{align}for all deformations $\epsilon\mapsto \Gamma_{\epsilon}$ such that $\delta q_t$ vanishes at $t = 0$ and $t= T$.
    \item The stochastic implicit Euler-Lagrange equations 
    \begin{align}
    \del q_t &= v_t\del X^0_t + \sum_{i=1}^kV_i(q_t)\del X^i_t\nonumber\\
    \del p_t &= \frac{\partial}{\partial q_t}\left(\L\del X_t^0 + \sum_{i = 1}^k\left(L_i - \dotp{p_t}{V_i(q_t)}\right)\del X^i_t\right)\nonumber\\
    \left(p_t - \pard{\L}{v_t}\right)\del X^0_t&=0.
\end{align}are satisfied by $\Gamma^{|T}_t = (q_t^{|T}, v_t^{|T}, p_t^{|T})$.
    \item The reduced semimartingale \[[q_t, v_t, p_t]_G \cong (x_t, u_t, y_t, \bar{\zeta}_t, \bar{\mu}_t)\]
    in the reduced Pontryagin bundle $\P Q/G\cong \P (Q/G)\oplus \alag\oplus \dalag$ is a critical point of the reduced action integral \begin{align}
\S_X^{\reduced}(x_t, u_t, y_t, \bar{\zeta}_t, \bar{\mu}_t) = &\int \ell(x_t, u_t, \bar{\zeta}_t)\del X^0_t + \sum_{i = 1}^k l_i(x_t) \del X^i_t + \dotp{y_t}{ \del x_t - u_t \del X^0_t - \sum_{i = 1}^kV_i^{\reduced}(x_t)\del X^i_t} \nonumber\\&+ \dotp{ \bar{\mu}_t}{ \del \bar{\xi}_t - \bar{\zeta}_t\del X^0_t - \sum_{i=1}^k \bar{\beta}_i(x_t) \del X^i_t}
\end{align}
for variations arbitrary variations $\delta u_t, \delta y_t, \delta \bar{\zeta}_t$ and $\delta \bar{\mu}_t$ and for variations
$\delta x^{|T}_t\oplus \delta^A( \bar{\xi}^{|T}_t)$ such that $\delta x^{|T} = 0$ at $t = 0$ and $t = T$ and \begin{align*}
    {\del (\delta^A\bar{\xi}^{|T}_t)} &= -\ad_{\bar{\eta}_t}{\del \bar{\xi}^{|T}_t} + {\Del \bar{\eta}_t} + {i_{\delta x^{|T}_t}\tilde{B}(x_t)}\del x_t
    \end{align*} where $\bar{\eta}$ is an arbitrary semimartingale in $\alag$ that vanishes at $t = 0$ and $t = T$.
    \item The horizontal stochastic implicit Lagrange-Poincaré equations
 \begin{align*}
     \Del y_t &= \frac{\partial}{\partial x_t} \left( \ell \del X^0_t + \sum_{i =1}^k\left(l_i -  \dotp{y_t^{|T}}{V_i^{\reduced}(x_t)} - \dotp{\bar{\mu}_t}{\bar{\beta}_i(x_t)}\right)\del X^i_t\right) - \dotp{\bar{\mu}_t}{i_{\del x_t}\tilde{B}(x_t)}\\
     \del x_t &= u_t\del X^0_t + \sum_{i = 1}^kV_i^{\mathrm{red}}(x_t)\del X^i_t\\
     \left(y_t - \frac{\partial \ell}{\partial u_t}\right)\del X^0_t &= 0
 \end{align*}
and the vertical stochastic implicit Lagrange-Poincaré equations
\begin{align*}
    \Del \bar{\mu}_t &= \ad^*_{\del \bar{\xi}_t} \bar{\mu}_t\\
    \del \bar{\xi}_t &= \bar{\zeta}_t\del X^0_t + \sum_{i = 1}^k\bar{\beta}_i(x_t)\del X^i_t\\
    \left(\bar{\mu}_t - \frac{\partial \ell}{\partial \bar{\zeta}_t}\right)\del X^0_t &= 0
\end{align*}
are satisfied by $(x_t^{|T}, u_t^{|T}, y_t^{|T}, \bar{\zeta}_t^{|T}, \bar{\mu}_t^{|T})$.
    \end{enumerate}
\end{theorem}
\begin{remark}
    In the deterministic case, that is, $X^0 = t$, $X^i = 0$, all semimartingales may be replaced by smooth curves. We also replace $\del\bar{\xi}_t$ by $[q(t), A(q(t), \dot{q}(t))]_Gdt$. Under these replacements, the stochastic implicit horizontal and vertical Lagrange-Poincaré equations agree with their deterministic counterparts.
\end{remark}
\subsubsection*{Coordinate expressions}

Suppose $Q$ is an $n$-manifold and $Q/G$ has dimension $r$. We choose a trivialization of $Q$ as $U\times G$, where $U\subseteq\R^r$ is an open set, and $G$ acts only on the second factor by left multiplication. Let $(x,g) = (x^{\alpha}, g^a)$ be an element of $U\times G$. The principal connection $A$ acts on a tangent vector $(\dot{x}, \dot{g})$ by \[A(\dot{x}, \dot{g}) = \Ad_g(A_{(x,e)}(\dot{x}, g^{-1}\dot{g})) = \Ad_g(A_e(x)\dot{x} + g^{-1}\dot{g}) =\Ad_g(A_e(x)\dot{x}) + \dot{g}g^{-1},\]
where $A_e(x)\dot{x} = A_{(x, e)}(\dot{x}, 0)$. Hence, $A(x,g) = \Ad_g(A_e(x)dx + g^{-1}dg)$. Let $\xi = g^{-1}\dot{g}$. Then \[\bar{\xi} := [(x,g), A(x,g,\dot{x}, \dot{g})]_G = [(x, e), A_e(x)\dot{x} + \xi]_G.\]Trivializing the adjoint bundle as $(U\times G\times \lag)/G\cong U\times \lag$ identifies $\bar{\xi}$ with $(x, A_e(x)\dot{x} + \xi)$. Let us write $A_e(x)$ locally as $A_{\alpha}^a(x)$. Then, with identify $\bar{\xi}$ with its second component and simply write $\bar{\xi}^a = \xi^a +A^{a}_{\alpha}\dot{x}^{\alpha}$. If we now use Stratonovich differentials, then this reads $\del\bar{\xi}^a = \del \xi^a + A^{a}_{\alpha}\del {x}^{\alpha}$, where $\del \bar{\xi}^a$, as a Stratonovich differential in $(U\times G\times \lag)/G\cong U\times \lag$ is given by $[(x_t, e), A_e(x_t) \del x_t + \del \xi_t ]$ and $\del \xi_t = g_t^{-1}\del g_t$.

\medskip

Let $C^{b}_{cd}$ denote the structure constants of the Lie algebra. Then the $\lag$-valued curvature 2-form $B$ is given locally by 
\[B^{b}_{\alpha\beta} = \pard{A^b_{\beta}}{x^{\alpha} } - \pard{A^b_{\alpha}}{x^{\beta} } - C^b_{cd}A^c_{\alpha}A^d_{\beta}.\]
Let the $G$-invariant vector fields $V_i$ be written locally as $V_i(x,g) = V_i^{\alpha}(x)\pard{}{x^{\alpha}} + g\beta^{a}(x)$, where $\beta^a:U\rightarrow \lag$ is a smooth map. Then the reduced vector fields $V_i^{\reduced}$ on $U$ given locally by $V_i^{\reduced}(x) = V_i^{\alpha}(x)\pard{}{x^{\alpha}}$. The sections $\bar{\beta}_i$ of the associated bundle are given by \[\bar{\beta}_i(x) = [(x,e), A(V_i(x,e))]_G.\] We think of these as maps from $U\rightarrow \lag$ given by $\bar{\beta}_i^a(x) = A^{a}_{\alpha}(x)V_i^{\alpha}(x) + \beta_i^{a}(x)$. 

\medskip

Following the local coordinate calculations done in Cendra, Marsden and Ratiu \cite{cendra1997lagrangian}, the local form of the horizontal stochastic implicit Lagrange-Poincaré equations is given by
\begin{align*}
    \del y_{\alpha_t} &= \pard{}{x^{\alpha}_t}\left(\ell\del X^0_t + \sum_{i = 1}^k l_i\del X^i_t\right) - \sum_{i = 1}^k\left(y_{\beta_t}\pard{V_i^{\beta}(x_t)}{x^{\alpha}_t}+\right.\\&+\left. \bar{\mu}_a\left(C^a_{bd}\bar{\beta}_i^b(x_t)A^d_{\alpha} +\pard{\beta_i^a(x_t)}{x^{\alpha}_t}\right)\right)\del X^i_t\\&+\bar{\mu}_a B^a_{\alpha\beta}\del x^{\beta}_t - \pard{\ell}{\bar{\zeta}^a}    C^a_{db}A^b_{\alpha}\del \bar{\xi}_t^d \\
    \del x^{\alpha}_t &= u^{\alpha}_t \del X^0_t + \sum_{i = 1}^k V_i^{\alpha}(x_t)\del X^i_t\\
    \left(y_{\alpha_t} - \pard{\ell}{x^{\alpha}_t}\right)\del X^0_t &= 0
\end{align*}
and the stochastic vertical implicit Lagrange-Poincaré equations are given by 
\begin{align*}
    \del \bar{\mu}_{b_t} &= \bar{\mu}_{a_t}C^a_{db}(\del \bar{\xi}^d_t -A^{d}_{\alpha}\del x^{\alpha}_t)\\
    \del\bar{\xi}^d_t &= \bar{\zeta}^d_t\del X^0_t + \sum_{i = 1}^k \bar{\beta}_i^d(x_t)\del X^i_t\\
    \left(\bar{\mu}_{b_t} - \pard{\ell}{\bar{\zeta}^b_t}\right)\del X^0_t &= 0.
\end{align*}
\begin{remark}
    If $X^0_t = t$ and $X^i_t = 0$ then $\bar{\mu} = \pard{\ell}{\bar{\zeta}}$ and $\del \bar{\xi}$ is replaced by the deterministic $\lag$-valued curve $\bar{\xi} = [q, A(q, \dot{q})]_G$. In this case, the local form of the reduced equations agree with the local form of the implicit Lagrange-Poincaré equations given in Yoshimura and Marsden \cite{yoma}.
\end{remark}
If the bundle $Q\rightarrow Q/G$ is equipped with a trivial connection in local coordinates, that is, $A = 0$, then the horizontal equations in local coordinates are given by
\begin{align*}
    \del y_{\alpha_t} &= \pard{}{x^{\alpha}_t}\left(\ell\del X^0_t + \sum_{i = 1}^k l_i\del X^i_t\right) - \sum_{i = 1}^k\left(y_{\beta_t}\pard{V_i^{\beta}(x_t)}{x^{\alpha}_t}+\right.\\&+\left. \bar{\mu}_a\pard{\beta_i^a(x_t)}{x^{\alpha}_t}\right)\del X^i_t\\
    \del x^{\alpha}_t &= u^{\alpha}_t \del X^0_t + \sum_{i = 1}^k V_i^{\alpha}(x_t)\del X^i_t\\
    \left(y_{\alpha_t} - \pard{\ell}{x^{\alpha}_t}\right)\del X^0_t &= 0
\end{align*}
and the vertical equations are given by 
\begin{align*}
    \del \bar{\mu}_{b_t} &= \bar{\mu}_{a_t}C^a_{db}(\del \bar{\xi}^d_t)\\
    \del\bar{\xi}^d_t &= \bar{\zeta}^d_t\del X^0_t + \sum_{i = 1}^k {\beta}_i^d(x_t)\del X^i_t\\
    \left(\bar{\mu}_{b_t} - \pard{\ell}{\bar{\zeta}^b_t}\right)\del X^0_t &= 0.
\end{align*}
We will call these equations the \textbf{stochastic Hamel's equations}. When $X^0_t = t$ and $X^i_t = 0$ then these equations correspond to Hamel's equations described in Marsden and Scheurle \cite{marsden1993reduced}.
\subsubsection*{Special Cases}
We now discuss four special cases.

\begin{enumerate}
\item Suppose $Q = G$. In this case, $Q/G$ is a point and hence the horizontal implicit stochastic Lagrange-Poincaré equations vanish. The $G$-invariant vector fields $V_i$ are of the form $V_i(g) = T_eL_g \beta_i$, where $\beta_i \in \lag$ is a fixed element, and $g\in G$. We also have $\alag\cong \lag$ and $\dalag \cong \dlag$. Then, the vertical implicit stochastic Lagrange-Poincaré equations are given by
\begin{align*}
    \del \mu_t &= \ad^*_{\del \xi_t}\mu_t\\
    \del \xi_t &= \zeta_t\del X^0_t + \sum_{i = 1}^k \beta_i \del X^i_t\\
    \left(\mu_t - \pard{\ell}{\zeta_t}\right)\del X^0_t &= 0.
\end{align*}
When $X^0_t = t$ and $X^i$ is a Brownian motion, these equations agree with the stochastic Euler-Poincaré equations given in Street and Takao \cite{street2023}. 

\item Assume that $G$ is an abelian group. From Equation \eqref{eq:covariantderivativeassociatedbundle}, the covariant derivative in the associated bundle agrees with the fibre derivative since the Lie bracket vanishes. Then, by Equation \eqref{eq:covariantderivativedualbundle}, it follows that the covariant derivative on the dual bundle agrees with the fibre derivative as well. Also, from the stochastic vertical implicit Lagrange-Poincaré equations, we see that $\Del \bar{\mu}_t = 0$. Writing $\bar{\mu}_t = [q_t, \mu_t]_G$, we see that $\mu_t$ is conserved along solutions. 

\medskip

As a special case, note that if $G =\{e\}$ then the vertical stochastic implicit Lagrange-Poincaré equations vanish and the horizontal stochastic implicit Lagrange-Poincaré equations agree with the stochastic implicit Euler-Lagrange equations.

\item We now assume that the noise vector fields  $V_i$ are horizontal. In this case, $\bar{\beta}_i = 0$, so the stochastic horizontal and vertical implicit Lagrange-Poincaré equations are given by

\begin{align*}
     \Del y_t &= \frac{\partial}{\partial x_t} \left( \ell \del X^0_t + \sum_{i =1}^k\left(l_i -  \dotp{y_t^{|T}}{V_i^{\reduced}(x_t)} \right)\del X^i_t\right) - \dotp{\bar{\mu}_t}{i_{\del x_t}\tilde{B}(x_t)}\\
     \del x_t &= u_t\del X^0_t + \sum_{i = 1}^kV_i^{\mathrm{red}}(x_t)\del X^i_t\\
     \left(y_t - \frac{\partial \ell}{\partial u_t}\right)\del X^0_t &= 0
 \end{align*}
and 
\begin{align*}
    \Del \bar{\mu}_t &= \ad^*_{\del \bar{\xi}_t} \bar{\mu}_t\\
    \del \bar{\xi}_t &= \bar{\zeta}_t\del X^0_t\\
    \left(\bar{\mu}_t - \frac{\partial \ell}{\partial \bar{\zeta}_t}\right)\del X^0_t &= 0
\end{align*}
respectively. As a result, if $X^0_t = t$ then the vertical equations are noise-free. 

\item Finally, suppose that $\beta_1, \cdots, \beta_k$ are $\Ad_G$-invariant elements of $\lag$, and let $V_i = (\beta_i)_Q$. Then $V_i$ is $G$-invariant, $V_i^{\reduced} = 0$ and $A(V_i) = \beta_i$. We also have $\bar{\beta}_i(x) = [q, \beta_i]_G$, where $x\in Q/G$ and $q\in \pi\inv(x)$. Note that this is independent of the choice of the representative $q$ since $\beta_i$ is $\Ad_G$-invariant. As a result $\frac{\partial}{\partial x_t}\dotp{\bar{\mu}_t}{\bar{\beta}_i(x_t)}$ vanishes. We also assume that $l_1, \cdots, l_k = 0$ and $X^0_t = t$. Then the horizontal equations are given by 
\begin{align*}
     \frac{Dy}{Dt} &= \frac{\partial\ell}{\partial x} - \dotp{\bar{\mu}_t}{i_{\dot{x}}\tilde{B}(x)} \\
     \dot {x} &= u \\
     y &= \frac{\partial \ell}{\partial u}
 \end{align*}
and the vertical equations are given by 
\begin{align*}
    \Del \bar{\mu}_t &= \ad^*_{\del \bar{\xi}_t} \bar{\mu}_t\\
    \del \bar{\xi}_t &= \bar{\zeta}_t dt + \sum_{i = 1}^k\bar{\beta}_i\del X^i_t\\
    \bar{\mu}_t &= \frac{\partial \ell}{\partial \bar{\zeta}_t}.
\end{align*}
This shows that we have a stochastic curvature-induced force in the horizontal equations.
\end{enumerate}

\section{Examples}

\subsection{Rigid Body with a Rotor}
Let us look at the example of a rigid body with a rotor aligned with its third principal axis. We will consider external stochastic forces on the rotor and the body, but for simplicity, we will assume that these forces do not affect the moments of inertia of the rotor and the body. 
\subsubsection*{The Deterministic Free Rigid Body with a Rotor}
Following Marsden \cite{marsden1} and Yoshimura and Marsden \cite{yoma}, we briefly recollect the setup of the deterministic problem. Let $I_1>I_2> I_3$ denote the principal rigid body moments of inertia, $K_1 = K_2$ denote the rotor moments of inertia about the first two principal axes of the rigid body and $K_3$ denote the moment of inertia of the rigid body about the third principal axis. Let $K$ be the diagonal matrix $K = \mathrm{diag}(K_1, K_2, K_3)$. The configuration space is $Q = S^1\times SO(3)$, where $G =SO(3)$ acts on the second factor by matrix multiplication. Then $Q/G = S^1$. The Lagrangian of this system is given by 
\begin{equation}\label{eq:Lagrangian_for_rigid_body_plus_rotor}\L(\theta, R, v_{\theta}, v_R) = \frac{1}{2}\left[\dotp{\Sigma}{I\Sigma} + \dotp{\Sigma + \vvec_{\theta}}{K(\Sigma + \vvec_{\theta}})\right],\end{equation}
where $\vvec_{\theta} = (0, 0, v_{\theta})^T\in \R^3$, $\Sigma\in \R^3$ is defined by $\hat\Sigma = R^{-1}v_R\in \mathfrak{so}(3)$ with the usual `hat' map identification of $\R^3$ with $\mathfrak{so}(3)$ given by $\hat{\Sigma}w = \Sigma\times w$, for any $w\in \R^3$.

\medskip

We let $A$ be a trivializing connection on the bundle $Q\rightarrow Q/G$, so that $TQ/G\cong TS^1\times \mathfrak{so}(3)$. Concretely, $A(\theta, R, v_{\theta}, v_R) = R^{-1}v_R \in \mathfrak{so}(3)$. Then, in terms the coordinates $(\theta, u, \Sigma)$ on $TS^1\times \mathfrak{so}(3)$, the reduced Lagrangian is given by 
\begin{equation}\label{eq:reduced_Lagrangian_rigid_body_plus_rotor}
    \ell(\theta, u, \Sigma) = \frac{1}{2}\left[\lambda_1\Sigma_1^2 + \lambda_2\Sigma_2^2 + I_3\Sigma_3^2 + K_3(\Sigma_3 + u)^2\right],
\end{equation}
with $\lambda_i = I_i + K_i$. The (deterministic) implicit horizontal Lagrange-Poincaré equations are given by 
\begin{align*}
    \dot{y} = 0, \:\:\dot{\theta} = u,\:\:y = \pard{\ell}{u} = K_3(\Sigma_3 + u)
\end{align*}
and the vertical implicit Lagrange-Poincaré equations are given by 
\begin{align*}
    \dot{\Pi} = \Pi\times \Omega, \:\:\Omega = \Sigma,\:\:\Pi = \pard{\ell}{\Sigma} = (\lambda_1\Sigma_1, \lambda_2\Sigma_2, I_3\Sigma_3 + K_3(\Sigma_3+u)).
\end{align*}
Note that the conjugate momentum $y = K_3(\Sigma_3 + u)$ is conserved.
\subsubsection*{Stochastic Perturbations of a Free Rigid Body}
Before considering a stochastic pertubation of the deterministic problem, we recall from Lázaro-Camí and Ortega \cite{lco2} and Arnaudon, De Castro and Holm \cite{holm_rigid_body}, the Stratonovich equations of motion describing a free rigid body under small random impacts. Let $\hat{\beta}_1, \hat{\beta}_2, \hat{\beta}_3$ be elements in the Lie algebra $\mathfrak{so}(3)$ and $X^1_t, X^2_t$ and $X^3_t$ be semimartingales. The body angular momentum of the stochastic free rigid body is given by
\begin{equation}\label{eq:stochastic_free_rigid_body}
    \del \Pi_t = (\Pi_t \times I^{-1}\Pi_t)dt + \sum_{i = 1}^3 (\Pi_t \times \beta_i)\del X^i_t,
\end{equation}
and the attitude is given by 
\[\del R_t = R_t \widehat{I^{-1} \Pi_t}dt + \sum_{i = 1}^3R_t\hat{\beta}_i\del X^i_t. \]
Note that Equation \eqref{eq:stochastic_free_rigid_body} is equivalent to
\begin{align*}
    \del \Pi_t &= \Pi_t \times \del \Omega_t\\
    \del \Omega_t &= \Sigma_tdt + \sum_{i = 1}^3 \beta_i\del X^i_t\\
    \Pi_t &= I\Sigma_t.
\end{align*}
These equations can be obtained by stochastic Euler-Poincaré reduction for the action
\[\Ac{}{R_t, v_{R_t}, p_{R_t}} = \int_0^T\left[ \frac{1}{2}\dotp{\Sigma_t}{I\Sigma_t}dt + \dotp{p_t}{\del R_t - v_{R_t}dt - \sum_{i = 1}^3 R_t\hat{\beta_i}\del X^i_t}\right],\]
where $\hat{\Sigma_t} = R_t^{-1}v_{R_t}$. Comparing this to the general structure of the stochastic Hamilton-Pontryagin action, we note that $\L =\frac{1}{2}\dotp{\Sigma}{I\Sigma}$, $X^0_t = t$, the noise Lagrangian $l_i$ are zero, and the noise vector fields are the left invariant vector fields corresponding to $\hat{\beta}_i$.

\subsubsection*{Stochastic Perturbations of a Rigid Body with a Rotor}
Let $\L$ be the Lagrangian given in Equation \eqref{eq:Lagrangian_for_rigid_body_plus_rotor} and $X^1_t$, $X^2_t$ and $X^3_t$ be semimartingales. Suppose $\hat{\beta}_1, \hat{\beta}_2$ and $\hat{\beta}_3$ are fixed elements of $\mathfrak{so}(3)$. On $Q = S^1\times SO(3)$ we consider the action
\begin{align*}\S(\theta_t, R_t, v_{\theta_t},v_{R_t}, p_{\theta_t}, p_{R_t}) &= \int_0^T \left[\L\:dt + \sum_{i = 1}^3L_i({\theta_t})\del X^i_t + \dotp{p_{\theta_t}}{\del \theta_t - v_{\theta_t}dt - \sum_{i = 1}^3 V_i(\theta_t)\del X^i_t}\right.\\
&\left.+ \dotp{p_{R_t}}{\del R_t - v_{R_t}dt - \sum_{i = 1}^3 R_t\hat{\beta_i}\del X^i_t},\right]
\end{align*}
where $L_1, L_2, L_3$ are smooth functions and $V_i$ is a vector field on $S^1$.
\begin{remark}
    The Lagrangians $L_1, L_2, L_3$ are interpreted as the contribution of the external force to the potential energy of the rotor, and the vector fields $V_1, V_2, V_3$ stochasticize the relation $\dot{\theta} = v_{\theta}$ to \[\del \theta_t = v_{\theta_t}dt + \sum_{i = 1}^3 V_i(\theta_t)\del X^i_t.\]The term \[\int_0^T\dotp{p_{R_t}}{\del R_t - v_{R_t}dt - \sum_{i = 1}^3 R_t\hat{\beta_i}\del X^i_t}\]accounts for random perturbation of the free rigid body. 
\end{remark}
We consider a trivializing connection $A$ on $S^1\times SO(3) \rightarrow S^1$. Then, the reduced action is given by
\begin{align*}\S^{\reduced}(\theta_t, u_t,\Sigma_t, y_t, \Pi_t) &= \int_0^T \left[\ell\:dt + \sum_{i = 1}^3L_i({\theta_t})\del X^i_t + \dotp{y_t}{\del \theta_t - u_tdt - \sum_{i = 1}^3 V_i(\theta_t)\del X^i_t}\right.\\
&\left.+ \dotp{\Pi_t}{\del \Omega - \Sigma_tdt - \sum_{i = 1}^3 \hat{\beta_i}\del X^i_t}\right],
\end{align*}
where $\ell$ is the reduced Lagrangian in Equation \eqref{eq:reduced_Lagrangian_rigid_body_plus_rotor} and $\del \Omega_t = R_t\inv\del R_t$. It follows that the stochastic implicit horizontal Lagrange-Poincaré equations are given by
\begin{align*}
    \del y_t &= \pard{}{\theta_t}\sum_{i = 1}^3\left(L_i(\theta_t) - y_t V_i(\theta_t)\right)\del X^i_t\\
    \del \theta_t &= u_tdt + \sum_{i = 1}^3 V_i(\theta_t)\del X^i_t\\
    y_t &= \pard{\ell}{u_t} = K_3(\Sigma_{3_t} + u_t),
\end{align*}
and the stochastic implicit vertical Lagrange-Poincaré equations are given by
\begin{align*}
    \del \Pi_t &= \Pi_t \times \del \Omega_t\\
    \del \Omega_t &= \Sigma_tdt + \sum_{i = 1}^3\beta_i\del X^i_t\\
    \Pi_t&= \pard{\ell}{\Sigma_t} = (\lambda_1\Sigma_{1_t}, \lambda_2\Sigma_{2_t}, I_3\Sigma_{3_t} + K_3(\Sigma_{3_t}+u_t)). 
\end{align*}
We observe that the conjugate momentum $y = K_3(\Sigma_3 + u)$ is no longer conserved in the stochastic case. 

\medskip

Let us mention three special cases.
\begin{enumerate}
    \item{\(L_i = 0\) and \(V_i = 0\)}: We consider $L_i = 0$ and $V_i = 0$. In this case, the horizontal equations are given by 
    \begin{align*}
    \dot{y}_t &= 0\\
    \dot{\theta}_t &= u_t\\
    y_t &= \pard{\ell}{u_t} = K_3(\Sigma_{3_t} + u_t).
\end{align*}
This implies the conjugate momentum $y = K_3(\Sigma_3 + u)$ is conserved, and the rotor angle of rotation evolves along differentiable trajectories. Physically, we interpret this as the case where the external noise only impacts the rigid body and not the rotor.
\item{\(\beta_i = 0\)}: In this case, the vertical equations are same as the vertical implicit Lagrange-Poincaré equations in the deterministic case. Physically, this models a system where the external noise only impacts the rotor and not the rigid body.
\item{\(V_i = 0\), \(L_i(\theta) =\theta\) and \(X^i_t\) is a Brownian motion}: In this case, the horizontal equations are given by 
\begin{align*}
    \del y_t &= \sum_{i = 1}^3\del B^i_t\\
    \dot \theta_t &= u_t\\
    y_t &= \pard{\ell}{u_t} = K_3(\Sigma_{3_t} + u_t),
\end{align*}
where $B^1, B^2$ and $B^3$ are independent Brownian motions. This implies that $y_t = B^1_t + B^2_t + B^3_t$. Hence $\E[y_t] = 0$ for all $t$. This shows that $y_t$ is a weakly conserved quantity in the sense of Lázaro-Camí and Ortega \cite[Definition 2.2]{lco1}, but not a conserved quantity. 
\end{enumerate}

\subsection{Charged Particle in a Magnetic Field with Stochastic Perturbations}
In this section we provide a Kaluza-Klein description of a stochastically perturbed charged particle in a magnetic field. 

\subsubsection*{The Deterministic Kaluza-Klein Description for a Charged Particle in a Magnetic Field}
The equation of motion for a charged particle (of unit mass) in a magnetic field $\Bvec$ is given by $\dot{\uvec} = \frac{e}{c}\uvec\times\Bvec$. Following Marsden and Ratiu \cite{marsden2}, we show that this can be viewed as a reduction of the geodesic flow on $Q_K = \R^3\times S^1$ under a certain metric. Here $S^1$ acts on the second factor by rotations.

\medskip

Let $G = S^1$ with its standard bi-invariant metric $\kappa$ and consider $\R^3$ with its standard metric given by the inner product $\dotp{\cdot}{\cdot}$. Let $\Avec$ be a vector in $\R^3$ and identify $\Avec$ with a 1-form $A$ on $\R^3$. Let \[\alpha = A + d\theta\] be a connection 1-form on the bundle $\pi:\R^3\times S^1\rightarrow \R^3$. Consider the metric on $Q_K$ given by \[g((\uvec_q, u_{\theta}),(\vvec_q, v_{\theta})) = \langle \uvec_q, \vvec_q\rangle+ \kappa(\alpha(\uvec_q, v_{\theta}), \alpha(\vvec_q, v_{\theta})).\]The Lagrangian for the geodesic flow on $(Q_K, g)$ is given by \[\L_{K}(\qvec, \theta, \vvec_q, v_{\theta}) = \frac{1}{2}\left(||\vvec_q||^2 + (\Avec\cdot\vvec_q + v_{\theta})^2\right).\] We will call it the \textbf{Kaluza-Klein Lagrangian}. Let $B = d\alpha = dA$ and identify $B$ with the vector $\Bvec = \grad\times\Avec$. The reduced curvature 2-form on $Q_K/S^1 \cong \R^3$ is identified with $B$ or  the vector $\Bvec$. Let $(\xvec,\uvec,\lambda)\in \R^3\times\R^3\times \R$ denote local coordinates on the bundle $T\R^3\oplus \R$. The reduced Lagrangian is \[\ell_K(\xvec, \uvec, \lambda) = \frac{1}{2}\left(||\uvec||^2 + \lambda^2\right).\]
The vertical implicit Lagrange-Poincaré equations are given by 
\[\dot{p}_{\theta} = 0,\:\:\chi = \lambda,\:\:p_{\theta} = \pard{\ell_K}{\lambda} = \lambda,\]
where $\chi = A\cdot \dot{q} + \dot{\theta}$. Since $\dot{p}_{\theta} = 0$, it follows that $p_{\theta}$ is a constant and we set $p_{\theta} = \frac{e}{c}$. The horizontal Lagrange-Poincaré equations are then given by 
\[\dot{\yvec} = \frac{e}{c}(\dot{\xvec}\times \Bvec),\:\:\dot{\xvec} = \uvec,\:\:\yvec = \pard{\ell_K}{\uvec} = \uvec,\]
which yields $\dot{\uvec} = \frac{e}{c}\uvec\times\Bvec$.

\subsubsection*{The Stochastic Case}

Suppose that $X^1, \cdots, X^k$ are arbitrary semimartingales, $L_1, \cdots, L_k\in \Cin(R^3)$ and $V_1, \cdots, V_k$ are vectors fields on $\R^3$. We consider the stochastic Hamilton-Pontryagin action functional 
\begin{align*}
    \Ac{}{\qvec_t, \theta_t, \uvec_{q_t}, u_{\theta_t}, \pvec_{q_t}, p_{\theta_t}} &= \int_0^T\left[ \L_K\:dt + \sum_{i = 1}^k L_i(\qvec_t)\del X^i_t +\dotp{\pvec_{q_t}}{\del \qvec_t - \uvec_{q_t}dt - \sum_{i = 1}^k V_i(\qvec_t)\del X^i_t} \right.\\&\left.+\dotp{p_{\theta_t}}{\del \theta_t - u_{\theta_t}dt}\right].
\end{align*}
Then, in terms of the Kaluza-Klein description, the reduced action corresponding to $\mathcal{S}$ is given by
\begin{align*}
    \mathcal{S}^{\reduced}(\xvec_t, \uvec_t, \lambda_t, \yvec_t, p_{\theta_t}) &= \int_0^T\left[ \ell_K\:dt + \sum_{i = 1}^k L_i(\xvec_t)\del X^i_t +\dotp{\yvec_t}{\del \xvec_t - \uvec_{t}dt - \sum_{i = 1}^k V_i(\xvec_t)\del X^i_t} \right.\\&\left.+\dotp{p_{\theta_t}}{\del \chi_t - \lambda_tdt}\right],
\end{align*}
where $\del \chi_t = \Avec\del \qvec_t + \del \theta_t$. The stochastic implicit vertical Lagrange-Poincaré equations are given by 
\begin{align*}
    \del p_{\theta_t} &=  0\\
    \dot{\chi}_t &= \lambda_t\\
    p_{\theta_t} &= \lambda_t
\end{align*}
which shows that $p_{\theta_t}$ is conserved. As before, we set $p_{\theta_t} = \frac{e}{c}$. Then the stochastic implicit horizontal Lagrange-Poincaré equations are given by 
\begin{align*}
    \del \yvec_t &= \frac{e}{c}(\del \xvec_t \times \Bvec)dt +  \sum_{i = 1}^k\pard{}{\xvec_t}\left(L_i(\xvec_t) - \yvec_t\cdot V_i(\xvec_t)\right)\del X^i_t\\
    \del \xvec_t &= \uvec_t dt + \sum_{i = 1}^k V_i(\xvec_t)\del X^i_t\\
    \yvec_t &= \uvec_t.
\end{align*}
Equivalently, we can solve for
\begin{align*}
    \del \uvec_t &=  \frac{e}{c}(\uvec_t \times \Bvec)dt +  \sum_{i = 1}^k\left(\frac{e}{c}(V_i(\xvec_t)\times \Bvec)+\pard{}{\xvec_t}\left(L_i(\xvec_t) - \uvec_t\cdot V_i(\xvec_t)\right)\right)\del X^i_t\\
    \del \xvec_t &= \uvec_t dt + \sum_{i = 1}^k V_i(\xvec_t)\del X^i_t,
\end{align*}
where the first equation represents the Lorentz force law with a stochastic perturbation. Note that if $k = 3$, $(X_t^1, X_t^2, X_t^3) = \mathbf{W}_t$, where $\mathbf{W}_t$ is a Brownian motion in $\R^3$, $L_i(x) = x^i$ and $V_i = 0$ then these equations become
\begin{align*}
    \del \uvec_t &=  \frac{e}{c}(\uvec_t \times \Bvec)dt +  \del \mathbf{W}_t\\
    \dot \xvec_t &= \uvec_t.
\end{align*}
Then, $(\E[\xvec_t], \E[\uvec_t])$ satisfies the equations for the charged particle in a magnetic field. 
\subsection*{Acknowledgements}
The author is thankful to Cristina Stoica and Tanya Schmah for their feedback and suggestions in preparing this manuscript.

\bibliographystyle{unsrt}
\bibliography{refer}

\appendix
\addcontentsline{toc}{section}{Appendix}
\section{Proof of Theorem \ref{thm:variation_of_Sred}}

\begin{theorem*}
Let $\Gamma_t = (x_t, u_t, y_t, \bar{\zeta}_t, \bar{\mu}_t)$ be an admissible semimartingale in $\P (Q/G)\oplus \alag\oplus \dalag$, $K$ be a regular coordinate ball in $\P (Q/G)\oplus \alag\oplus \dalag$, $\tau_K^h$ be the hitting time of $\Gamma$ for $K$ and $\tau_{K}^{(h,e)}$ be the exit time of $\Gamma_{t + \tau_K^k}$ for $K$. Suppose $\epsilon \mapsto\Gamma^{|T}_{\epsilon,t} =(x^{|T}_{\epsilon,t}, u^{|T}_{\epsilon,t},y^{|T}_{\epsilon,t},  \bar{\zeta}^{|T}_{\epsilon,t}, \bar{\mu}^{|T}_{\epsilon,t})$ is a deformation of $\Gamma_t^{|T} = (x^{|T}_t, u^{|T}_t, y^{|T}_t, \bar{\zeta}^{|T}_t, \bar{\mu}^{|T}_t)$ such that the variations $\delta u^{|T}_t, \delta \bar{\zeta}^{|T}_t, \delta  y^{|T}_t$ and $\delta \bar{\mu}^{|T}_t$ are arbitrary and the variations $\delta x^{|T}_t\oplus \delta^A( \bar{\xi}^{|T}_t)$ satisfy $\delta x^{|T} = 0$ at $t = 0$ and $t = T$ and \begin{align*}
    {\del (\delta^A\bar{\xi}^{|T}_t)} &= -\ad_{\bar{\eta}_t}{\del \bar{\xi}^{|T}_t} + {\Del \bar{\eta}_t} + {i_{\delta x^{|T}_t}\tilde{B}(x_t)}\del x_t
    \end{align*} in the sense of Remark \ref{rem:stochastic_covariant_variations}. Here $\bar{\eta}$ is an arbitrary semimartingale in $\alag$ that vanishes at $t = 0$ and $t = T$. Then
\begingroup
\allowdisplaybreaks
\begin{align}\label{eq:variationinSred_app}
    \D \S^{\mathrm{red}}_X(x_{t}, u_{t}, y_{t}, \bar{\zeta}_t,\bar{\mu}_{t})&=
    \int_0^T \left\langle \frac{\partial}{\partial x_t^{|T}} \left( \ell \del X^0_t + \sum_{i =1}^k\left(l_i -  \dotp{y_t^{|T}}{V_i^{\reduced}(x_t^{|T})} - \dotp{\bar{\mu}_t^{|T}}{\bar{\beta}_i(x_t^{|T})}\right)\del X^i_t\right)\right.\nonumber\\
    &\left. - \dotp{\bar{\mu}_t}{i_{\del x_t}\tilde{B}(x_t)} - \Del y_t^{|T},\delta x_t^{|T}\right\rangle +\int_0^T\left\langle\frac{\partial \ell}{\partial u^{|T}_t} - y^{|T}_t, K_{\grad^{T(Q/G)}}\delta u^{|T}_t\right\rangle \del X^0_t\nonumber\\
    &+ \int_0^T\left\langle\frac{\partial \ell}{\partial {\bar{\zeta}}^{|T}_t} - \bar{\mu}^{|T}_t, K_{\grad^{\alag}}\delta {\bar{\zeta}}_t\right\rangle \del X^0_t\nonumber\\&+\int_0^T \langle K_{\grad^{T^*(Q/G)}} \delta y^{|T}_t, \del x^{|T}_t - u^{|T}_t\del X^0_t - \sum_{i = 1}^k V_i^{\mathrm{red}}(x^{|T}_t)\del X^i_t\rangle\nonumber\\ &+ \int_0^T\langle K_{\grad^{\dalag}}\delta \bar{\mu}^{|T}_t, \del \bar{\xi}^{|T}_t - \bar{\zeta}^{|T}_t\del X^0_t - \sum_{i = 1}^k\bar{\beta}_i(x^{|T}_t)\del X^i_t\rangle\nonumber\\&+\int_0^T \langle -\Del\bar{\mu}^{|T}_t+\ad^*_{\del \bar{\xi}^{|T}_t}\bar{\mu}^{|T}_t, \bar{\eta}_t\rangle
 \end{align}
 \endgroup
Consequently, $ \D \S^{\mathrm{red}}_X(x_{t}, u_{t},y_{t}, \bar{\zeta}_{t},  \bar{\mu}_{t}) = 0$ holds if and only if the \textbf{horizontal stochastic Lagrange-Poincaré equations} 
 \begin{align*}
     \Del y_t &= \frac{\partial}{\partial x_t} \left( \ell \del X^0_t + \sum_{i =1}^k\left(l_i -  \dotp{y_t}{V_i^{\reduced}(x_t)} - \dotp{\bar{\mu}_t}{\bar{\beta}_i(x_t)}\right)\del X^i_t\right) - \dotp{\bar{\mu}_t}{i_{\del x_t}\tilde{B}(x_t)} \\
     \del x_t &= u_t\del X^0_t + \sum_{i = 1}^kV_i^{\mathrm{red}}(x_t)\del X^i_t\\
     \left(y_t - \frac{\partial \ell}{\partial u_t}\right)\del X^0_t &= 0
 \end{align*}
and the \textbf{vertical stochastic implicit Lagrange-Poincaré equations}
\begin{align*}
    \Del \bar{\mu}_t &= \ad^*_{\del \bar{\xi}_t} \bar{\mu}_t\\
    \del \bar{\xi}_t &= \bar{\zeta}_t\del X^0_t + \sum_{i = 1}^k\bar{\beta}_i(x_t)\del X^i_t\\
    \left(\bar{\mu}_t - \frac{\partial \ell}{\partial \bar{\zeta}_t}\right)\del X^0_t &= 0.
\end{align*}
are satisfied by $(x_t^{|T}, u_t^{|T}, y_t^{|T}, \bar{\zeta}_t^{|T}, \bar{\mu}_t^{|T})$.
\end{theorem*}
\begin{proof}
   
    We have\begin{align}\label{eq:variations_step_0}
        \D \S^{\mathrm{red}}_X(x_{t}, u_{t}, y_{t}, \bar{\zeta}_t,\bar{\mu}_{t}) &= \D\int_0^T\ell(x_{t}, u_{t}, \bar{\zeta}_{t})\del X^0_t + \sum_{i = 1}^k \D \int_0^Tl_i(x_{t}) \del X^i_t+\D \int_0^T\langle y_{t}, \del x_{t}\rangle \nonumber\\&- \D\int_0^T \langle y_{t}, u_{t} \rangle \del X^0_t - \sum_{i = 1}^k\int_0^T\langle y_t,V_i^{\mathrm{red}}(x_{t})\rangle \del X^i_t+\D\int_0^T\langle   \bar{\mu}_{t}, \del \bar{\xi}_{t}\rangle \nonumber\\&- \D\int_0^T\left\langle \bar{\mu}_{t},\bar{\zeta}_{t}dt + \sum_{i=1}^k \bar{\beta}_i(x_t) \del X^i_t\right\rangle\nonumber.\\
        &= \D\int_0^T\ell(x^{|T}_{t}, u^{|T}_{t}, \bar{\zeta}^{|T}_{t})\del X^0_t + \sum_{i = 1}^k \D \int_0^Tl_i(x^{|T}_{t}) \del X^i_t +\D \int_0^T\langle y^{|T}_{t}, \del x^{|T}_{t}\rangle \nonumber\\&- \D\int_0^T \langle y^{|T}_{t}, u^{|T}_{t} \rangle \del X^0_t - \sum_{i = 1}^k\int_0^T\langle y^{|T}_t,V_i^{\mathrm{red}}(x^{|T}_{t})\rangle \del X^i_t+\D\int_0^T\langle   \bar{\mu}^{|T}_{t}, \del \bar{\xi}^{|T}_{t}\rangle \nonumber\\&- \D\int_0^T\left\langle \bar{\mu}^{|T}_{t},\bar{\zeta}^{|T}_{t}dt + \sum_{i=1}^k \bar{\beta}_i(x^{|T}_t) \del X^i_t\right\rangle\nonumber.\\
        &=\D\int_0^T \ell(x^{|T}_{t}, u^{|T}_{t}, \bar{\zeta}^{|T}_{t})\del X^0_t + \sum_{i = 1}^k \D \int_0^Tl_i(x^{|T}_{t}) \del X^i_t \nonumber\\&+\D \int_0^T\langle y^{|T}_{t}, \del x^{|T}_{t}\rangle - \left(\D\int_0^T \langle y^{|T}_{t}, u^{|T}_{t} \rangle \del X^0_t \right.\nonumber\\&+ \left.\sum_{i = 1}^k\D\int_0^T\langle y^{|T}_t,V_i^{\mathrm{red}}(x^{|T}_{t})\rangle \del X^i_t\right)+\D\int_0^T\langle   \bar{\mu}^{|T}_{t}, \del \bar{\xi}^{|T}_{t}\rangle \nonumber\\&- \D\int_0^T\left\langle \bar{\mu}^{|T}_{t},\bar{\zeta}^{|T}_{t} \del X^0_t + \sum_{i=1}^k \bar{\beta}_i(x^{|T}_t) \del X^i_t\right\rangle.
    \end{align}
    We break down the calculation of variations of the terms in several steps:
    \begin{enumerate}
        \item \textbf{The terms $\D\int_0^T \ell(x^{|T}_{t}, u^{|T}_{t}, \bar{\zeta}^{|T}_{t})\del X^0_t+ \sum_{i = 1}^k \D \int_0^Tl_i(x^{|T}_{t}) \del X^i_t$}
        
        \medskip

        By Theorem \ref{thm:variation_vector_bundle}, we have
        \begin{align}\label{eq:variation_step_1}
            &\D\int_0^T \ell(x^{|T}_{t}, u^{|T}_{t}, \bar{\zeta}^{|T}_{t})\del X^0_t+\sum_{i = 1}^k \D \int_0^Tl_i(x^{|T}_{t}) \del X^i_t\nonumber\\ &= \int_0^T \dotp{d\ell(x^{|T}_t, u^{|T}_t, \bar{\zeta}^{|T}_t)}{(\delta x^{|T}_t, \delta u^{|T}_t, \delta \bar{\zeta}^{|T}_t)}\del X^0_t+ \sum_{i = 1}^k  \int_0^T\dotp{dl_i(x^{|T}_{t})}{\delta x_t^{|T}} \del X^i_t\nonumber\\
            &= \int_0^T \left(\dotp{\frac{\partial \ell}{\partial x^{|T}_t}}{\delta x^{|T}_t} + \int_0^T \dotp{\frac{\partial \ell}{\partial u^{|T}_t}}{K_{\nabla^{T(Q/G)}}\delta u^{|T}_t} \right.\nonumber\\&+\left. \dotp{\frac{\partial \ell}{\partial \bar{\zeta}^{|T}_t}}{K_{\nabla^{T(Q/G)}}\delta \bar{\zeta}^{|T}_t}\right)\del X^0_t + \sum_{i = 1}^k  \int_0^T\dotp{\frac{\partial l_i}{\partial x_t^{|T}}}{\delta x_t^{|T}} \del X^i_t.
        \end{align}
        \item \textbf{The term $\D \int_0^T\langle y^{|T}_{t}, \del x^{|T}_{t}\rangle$}
        
        \medskip

        We recognize this term as $\D\int_0^T \mathcal{G}_{\P (Q/G)}\del (x^{|T}_t, u^{|T}_t, y^{|T}_t)$. Since $\mathcal{G}_{\P (Q/G)}(\delta x^{|T}, \delta u^{|T}_t, \delta y^{|T} ) = \dotp{y^{|T}_t}{\delta x^{|T}_t}$ and $\delta x^{|T}_t$ vanishes at $t = 0$ and $t = T$, by Lemma \ref{lem:variationsofsemimartginales}, we have
        \begin{align*}
            \D \int_0^T\langle y^{|T}_{t}, \del x^{|T}_{t}\rangle &= \int_0^T i_{(\delta x^{|T}_t, \delta u^{|T}_t, \delta y^{|T}_t)}\d \mathcal{G}_{\P(Q/G)}\del (x^{|T}_t, u_t^{|T}, y^{|T}_t) \\&+ \dotp{\mathcal{G}_{\P (Q/G)} (x^{|T}_T, u_T^{|T}, y^{|T}_T)}{(\delta x^{|T}_T, \delta u^{|T}_T, \delta y^{|T}_T)} - \dotp{\mathcal{G}_{\P (Q/G)} (x^{|T}_0, u_0^{|T}, y^{|T}_0)}{(\delta x^{|T}_0, \delta u^{|T}_0, \delta y^{|T}_0)}\\
            &=\int_0^T i_{(\delta x^{|T}_t, \delta u^{|T}_t, \delta y^{|T}_t)}\d \mathcal{G}_{\P(Q/G)}\del (x^{|T}_t, u_t^{|T}, y^{|T}_t).
        \end{align*}
    Thus, we need to determine $i_{(\delta x^{|T}_t, \delta u^{|T}_t, \delta y^{|T}_t)}\d \mathcal{G}_{\P(Q/G)}$. To do this, we will use Corollary \ref{cor:importantcorollary}. 
    
    \medskip
    
    Let $(x, u ,y)$ be a point in $\P (Q/G)$. By definition, $\mathcal{G}_{\P (Q/G)}(x, u , y) = \dotp{y}{dx}$, so that $\d\mathcal{G}_{\P (Q/G)} = dy \wedge dx$. Suppose $(w_x, w_u, w_y)$ and $(\tilde{w}_x, \tilde{w}_u, \tilde{w}_y)$ are tangent vectors to $\P (Q/G)$ at $(x, u ,y)$. Let $\gamma(t) = (x(t), u(t), y(t))$ be a curve in $\P (Q/G)$ such that $\gamma(0) =(x, u ,y) $ and $\dot{\gamma}(0) = (w_x, w_u, w_y)$. Let $\gamma_{\epsilon}(t) = (x_{\epsilon}(t), u_{\epsilon}(t), y_{\epsilon}(t))$ be a deformation of $\gamma(t)$ with $\delta \gamma(0) =(\tilde{w}_x, \tilde{w}_u, \tilde{w}_y)$. By Corollary \ref{cor:importantcorollary}
    
\begin{align*}
    \frac{d}{d\epsilon}\Big|_{\epsilon = 0} \langle \mathcal{G}_{\P(Q/G)}(\gamma_{\epsilon}(t)), \dot{\gamma}_{\epsilon}(t)\rangle = \langle i_{\delta \gamma(t)}\d\mathcal{G}_{\P(Q/G)}, \dot {\gamma}(t)\rangle + \frac{d}{dt}\langle\mathcal{G}_{\P(Q/G)}(\gamma(t)), \delta\gamma(t)\rangle.
\end{align*}which shows that,
\begin{align*}
    \langle i_{\delta \gamma(t)}\d\mathcal{G}_{\P(Q/G)}, \dot {\gamma}(t)\rangle =  \frac{d}{d\epsilon}\Big|_{\epsilon = 0} \langle \mathcal{G}_{\P(Q/G)}(\gamma_{\epsilon}(t)), \dot{\gamma}(t)\rangle - \frac{d}{dt}\langle \mathcal{G}_{\P(Q/G)}(\gamma(t)), \delta\gamma(t)\rangle
\end{align*}
Expressing the right side in terms of coordinates yields,
\begin{align*}
    \langle i_{\delta \gamma(t)}\d\mathcal{G}_{\P(Q/G)}, \dot {\gamma}(t)\rangle = \frac{d}{d\epsilon}\Big|_{\epsilon = 0} \langle y_{\epsilon}(t), \dot{x}_{\epsilon}(t)\rangle -  \frac{d}{dt}\langle y(t), \delta x(t)\rangle.
\end{align*}
In terms of the covariant derivatives on $T(Q/G)$ and $T^*(Q/G)$, we have
\begin{align*}
    \frac{d}{d\epsilon}\Big|_{\epsilon = 0} \langle y_{\epsilon}(t), \dot{x}_{\epsilon}(t)\rangle -  \frac{d}{dt}\langle y(t), \delta x(t)\rangle &= \left\langle \Deldeleps y_{\epsilon}(t), \dot{x}(t) \right\rangle +  \left\langle y(t), \Deldeleps\dot{x}_{\epsilon}(t) \right\rangle \\&- \left\langle \frac{D}{Dt}y(t), \delta x(t)\right\rangle - \left\langle y(t), \frac{D}{Dt}\delta x(t)\right\rangle
\end{align*}
But by equality of mixed partial derivatives
\begin{align*}
    \left\langle y(t), \frac{D}{D\epsilon}\dot{x}(t) \right\rangle - \left\langle y(t), \frac{D}{Dt}\frac{d}{d\epsilon} x(t)\right\rangle &=  \left\langle y_{\epsilon}(t), K_{\grad^{T(Q/G)}}\frac{d^2x(t)}{d\epsilon d t} \right\rangle - \left\langle y(t), K_{\grad^{T(Q/G)}}\frac{d^2x(t)}{dtd\epsilon}\right\rangle\\
    &= 0,
\end{align*}
which implies
\begin{align*}\frac{d}{d\epsilon}\Big|_{\epsilon = 0} \langle y(t), \dot{x}_{\epsilon}(t)\rangle -  \frac{d}{dt}\langle y(t), \delta x(t)\rangle &= \left\langle \Deldeleps y_{\epsilon}(t), \dot{x}(t) \right\rangle -
\left\langle \frac{D}{Dt}y(t), \delta x(t)\right\rangle\\
&= \left\langle K_{\grad^{T^*(Q/G)}} \delta y(t), \dot{x}(t) \right\rangle - \left\langle K_{\grad^{T^*(Q/G)}}\dot{y}(t), \delta x(t)\right\rangle.
\end{align*}
Evaluating at $t = 0$ gives
\begin{align*}
    i_{(\tilde{w}_x, \tilde{w}_u, \tilde{w}_y)}\d\mathcal{G}_{\P (Q/G)}(w_x, w_u, w_y) &= \left(\frac{d}{d\epsilon}\Big|_{\epsilon = 0} \langle y(t), \dot{x}_{\epsilon}(t)\rangle -  \frac{d}{dt}\langle y(t), \delta x(t)\rangle\right)_{t = 0}\\
    &= \left(\left\langle K_{\grad^{T^*(Q/G)}} \delta y(t), \dot{x}(t) \right\rangle - \left\langle K_{\grad^{T^*(Q/G)}}\dot{y}(t), \delta x(t)\right\rangle\right)_{t = 0}\\
    &= \dotp{K_{\nabla^{T^*(Q/G)}}\tilde{w}_y}{w_x} - \dotp{K_{\nabla^{T^*(Q/G)}}{w}_y}{\tilde{w}_x}.
\end{align*}
    Since this holds for arbitrary tangent vectors $(w_x, w_u, w_y)$ and $(\tilde{w}_x, \tilde{w}_u, \tilde{w}_y)$ at $(x, u ,y)$, we have
    \begin{align*}i_{(\delta x^{|T}_t, \delta u^{|T}_t, \delta y^{|T}_t)}\d \mathcal{G}_{\P(Q/G)}\del (x^{|T}_t, u_t^{|T}, y^{|T}_t) &= \dotp{K_{\nabla^{T^*(Q/G)}}\delta y^{|T}_t}{\del x^{|T}_t} - \dotp{K_{\nabla^{T^*(Q/G)}}\del y^{|T}_t}{\delta x_t^{|T}}\\
    &=\dotp{K_{\nabla^{T^*(Q/G)}}\delta y^{|T}_t}{\del x^{|T}_t} - \dotp{\Del y^{|T}_t}{\delta x_t^{|T}}.
    \end{align*}
    As a result,
    \begin{equation}\label{eq:variation_step_2}
        \D \int_0^T\langle y^{|T}_{t}, \del x^{|T}_{t}\rangle = \int_0^T\dotp{K_{\nabla^{T^*(Q/G)}}\delta y^{|T}_t}{\del x^{|T}_t} - \int_0^T\dotp{\Del y^{|T}_t}{\delta x_t^{|T}}.
    \end{equation}
    \item\textbf{The terms $\D\int_0^T \langle y^{|T}_{t}, u^{|T}_{t} \rangle \del X^0_t + \sum_{i = 1}^k\D\int_0^T\langle y^{|T}_t,V_i^{\mathrm{red}}(x^{|T}_{t})\rangle \del X^i_t$}

    \medskip

    We use the connectors corresponding to the covariant derivatives on $T(Q/G)$ and $T^*(Q/G)$ to obtain
    \begin{align*}
      &\D\int_0^T \langle y^{|T}_{t}, u^{|T}_{t} \rangle \del X^0_t \\
      &= \int_0^T\dotp{K_{\nabla^{T^*(Q/G)}}\delta y^{|T}_t}{u_t^{|T}}\del X^0_t + \int_0^T\dotp{ y^{|T}_t}{K_{\nabla^{T(Q/G)}}\delta u_t^{|T}}\del X^0_t
    \end{align*}
    To evaluate $\D\int_0^T\langle y^{|T}_t,V_i^{\mathrm{red}}(x^{|T}_{t})\rangle \del X^i_t$, we apply Theorem \ref{thm:variation_vector_bundle} to the maps $(x,y)\in T^*(Q/G) \mapsto \dotp{y}{V_i^{\reduced}(x)}\in \R$. This gives,
    \begin{align*}
        \D\int_0^T\langle y^{|T}_t,V_i^{\mathrm{red}}(x^{|T}_{t})\rangle \del X^i_t &= \int_0^T \dotp{\frac{\partial}{\partial x_t^{|T}}\dotp{y_t^{|T}}{V_i^{\reduced}(x_t^{|T})}}{\delta x_t^{|T}}\del X^i_t \\&+ \int_0^T \dotp{\frac{\partial}{\partial y_t^{|T}}\dotp{y_t^{|T}}{V_i^{\reduced}(x_t^{|T})}}{K_{\nabla^{T^*(Q/G)}}\delta y_t^{|T}}\del X^i_t\\
        &= \int_0^T \dotp{\frac{\partial}{\partial x_t^{|T}}\dotp{y_t^{|T}}{V_i^{\reduced}(x_t^{|T})}}{\delta x_t^{|T}}\del X^i_t \\&+ \int_0^T \dotp{K_{\nabla^{T^*(Q/G)}}\delta y_t^{|T}}{V_i^{\reduced}(x_t^{|T})}\del X^i_t.
    \end{align*}
    Therefore
    \begin{align}\label{eq:variation_step_3}
         &\D\int_0^T \langle y^{|T}_{t}, u^{|T}_{t} \rangle \del X^0_t + \sum_{i = 1}^k\D\int_0^T\langle y^{|T}_t,V_i^{\mathrm{red}}(x^{|T}_{t})\rangle \del X^i_t\nonumber\\
      &= \int_0^T\dotp{K_{\nabla^{T^*(Q/G)}}\delta y^{|T}_t}{u_t^{|T}\del X^0_t + \sum_{i = 1}^k V_i^{\mathrm{red}}(x_t^{|T})\del X^i_t}\nonumber\\
      &+ \int_0^T\dotp{ y^{|T}_t}{K_{\nabla^{T(Q/G)}}\delta u_t^{|T}}\del X^0_t+\sum_{i = 1}^k\int_0^T \dotp{\frac{\partial}{\partial x_t^{|T}}\dotp{y_t^{|T}}{V_i^{\reduced}(x_t^{|T})}}{\delta x_t^{|T}}\del X^i_t.  
    \end{align}
    \item \textbf{The term $\D\int_0^T\langle   \bar{\mu}^{|T}_{t}, \del \bar{\xi}^{|T}_{t}\rangle$}

    \medskip

    Using the computations carried out in Section \ref{sec:covariant_variations_adjoint_bundle}, we have 
    \begin{align}
        \D\int_0^T\langle   \bar{\mu}^{|T}_{t}, \del \bar{\xi}^{|T}_{t}\rangle &= \int_0^T \dotp{K_{\nabla^{\dalag}}\delta\bar{\mu}^{|T}_t}{\del\bar{\xi}^{|T}_t}
-\int_0^T \dotp{\ad^*_{\bar{\eta}_t}\bar{\mu}^{|T}_t}{\del \bar{\xi}^{|T}}\nonumber\\&+\int_0^T \dotp{\bar{\mu}^{|T}_t}{i_{\delta x^{|T}_t}\tilde{B}(x^{|T}_t)}\del x^{|T}_t- \int_0^T \dotp{\Del \bar{\mu}^{|T}_t} {\bar{\eta}_t}
    \nonumber\\&+ \dotp{\bar{\mu}^{|T}_T}{\bar{\eta}_T} - \dotp{\bar{\mu}^{|T}_0}{\bar{\eta}_0}.\nonumber
    \end{align}
Since $\bar{\eta}$ vanishes at $t = 0$ and $t = T$, it follows that $\dotp{\bar{\mu}^{|T}_T}{\bar{\eta}_T} = \dotp{\bar{\mu}^{|T}_0}{\bar{\eta}_0} = 0$. Hence,
\begin{align}\label{eq:variations_Step_4}
        \D\int_0^T\langle   \bar{\mu}^{|T}_{t}, \del \bar{\xi}^{|T}_{t}\rangle &= \int_0^T \dotp{K_{\nabla^{\dalag}}\delta\bar{\mu}^{|T}_t}{\del\bar{\xi}^{|T}_t}
-\int_0^T \dotp{\ad^*_{\bar{\eta}_t}\bar{\mu}^{|T}_t}{\del \bar{\xi}^{|T}}\nonumber\\&+\int_0^T \dotp{\bar{\mu}^{|T}_t}{i_{\delta x^{|T}_t}\tilde{B}(x^{|T}_t)}\del x^{|T}_t- \int_0^T \dotp{\Del \bar{\mu}^{|T}_t} {\bar{\eta}_t}
    \nonumber\\
    &= \int_0^T \dotp{K_{\nabla^{\dalag}}\delta\bar{\mu}^{|T}_t}{\del\bar{\xi}^{|T}_t}
+\int_0^T \dotp{\ad^*_{\del\bar{\xi}^{|T}_t}\bar{\mu}^{|T}_t}{\bar{\eta}_t}\nonumber\\&-\int_0^T \dotp{\dotp{\bar{\mu}^{|T}_t}{i_{\del x^{|T}_t}\tilde{B}(x^{|T}_t)}}{\delta x^{|T}_t}- \int_0^T \dotp{\Del \bar{\mu}^{|T}_t} {\bar{\eta}_t}\nonumber\\
&= \int_0^T \dotp{K_{\nabla^{\dalag}}\delta\bar{\mu}^{|T}_t}{\del\bar{\xi}^{|T}_t}
-\int_0^T \dotp{\Del \bar{\mu}_t^{|T} -\ad^*_{\del\bar{\xi}^{|T}_t}\bar{\mu}^{|T}_t}{\bar{\eta}_t}\nonumber\\&-\int_0^T \dotp{\dotp{\bar{\mu}^{|T}_t}{i_{\del x^{|T}_t}\tilde{B}(x^{|T}_t)}}{\delta x^{|T}_t},
    \end{align}
    where \[\int\dotp{\ad^*_{\del\bar{\xi}^{|T}_t}\bar{\mu}^{|T}_t}{\bar{\eta}_t} = -\int\dotp{\ad^*_{\bar{\eta}_t}\bar{\mu}^{|T}_t}{\del\bar{\xi}^{|T}_t}\]and\[\int \dotp{\dotp{\bar{\mu}^{|T}_t}{i_{\del x^{|T}_t}\tilde{B}(x^{|T}_t)}}{\delta x^{|T}_t} = -\int\dotp{\bar{\mu}^{|T}_t}{i_{\delta x^{|T}_t}\tilde{B}(x^{|T}_t)}\del x^{|T}_t.\]
    
    \item \textbf{The terms $\D\int_0^T\left\langle \bar{\mu}^{|T}_{t},\bar{\zeta}^{|T}_{t} \del X^0_t + \sum_{i=1}^k \bar{\beta}_i(x^{|T}_t) \del X^i_t\right\rangle$}

    \medskip

    We proceed by applying Equation \eqref{eq:covariant_Derivative_dual_bundle_connector} to yield 
    \begin{align*}
        \D\int_0^T\left\langle \bar{\mu}^{|T}_{t},\bar{\zeta}^{|T}_{t} \del X^0_t \right\rangle &= \int_0^T\dotp{K_{\nabla^{\dalag}}\delta\bar{\mu}_t^{|T}}{\bar{\zeta}_t^{|T}}\del X^0_t \\&+ \int_0^T\dotp{\bar{\mu}^{|T}_t}{K_{\nabla^{\alag}}\delta\bar{\zeta}_t^{|T}}\del X^0_t.
    \end{align*}
    Next, by applying Theorem \ref{thm:variation_vector_bundle} to the maps $\bar{\mu}_x\in \dalag\mapsto \dotp{\bar{\mu}_x}{\bar{\beta}_i(x)}$, we obtain
    \begin{align*}
        \D\int_0^T\dotp{\bar{\mu}_t^{|T}}{\bar{\beta}_i(x_t^{|T})}\del X^i_t &= \int_0^T \dotp{\frac{\partial}{\partial x_t^{|T}}\dotp{\bar{\mu}_t^{|T}}{\bar{\beta}_i(x_t^{|T})}}{\delta x_t^{|T}} + \int_0^T\dotp{\frac{\partial}{\partial y_t^{|T}}\dotp{\bar{\mu}_t^{|T}}{\bar{\beta}_i(x_t^{|T})}}{K_{\nabla^{\dalag}}\delta \bar{\mu}_t^{|T}}\\
        &= \int_0^T \dotp{\frac{\partial}{\partial x_t^{|T}}\dotp{\bar{\mu}_t^{|T}}{\bar{\beta}_i(x_t^{|T})}}{\delta x_t^{|T}} + \int_0^T\dotp{K_{\nabla^{\dalag}}\delta \bar{\mu}_t^{|T}}{\bar{\beta}_i(x_t^{|T})}.
    \end{align*}
    Consequently, 
    \begin{align}\label{eq:variations_step_5}
        \D\int_0^T\left\langle \bar{\mu}^{|T}_{t},\bar{\zeta}^{|T}_{t} \del X^0_t + \sum_{i=1}^k \bar{\beta}_i(x^{|T}_t) \del X^i_t\right\rangle &= \int_0^T\dotp{K_{\nabla^{\dalag}}\delta\bar{\mu}_t^{|T}}{\bar{\zeta}_t^{|T}}\del X^0_t \nonumber\\&+ \int_0^T\dotp{\bar{\mu}^{|T}_t}{K_{\nabla^{\alag}}\delta\bar{\zeta}_t^{|T}}\del X^0_t\nonumber\\&+\left(\sum_{i = 1}^k\int_0^T\dotp{K_{\nabla^{\dalag}}\delta\bar{\mu}_t^{|T}}{\bar{\beta}_i(x_t^{|T})}\del X^i_t\right.\nonumber\\
        &+\int_0^T \dotp{\frac{\partial}{\partial x_t^{|T}}\dotp{\bar{\mu}_t^{|T}}{\bar{\beta}_i(x_t^{|T})}}{\delta x_t^{|T}}.
    \end{align}
    \end{enumerate}
    From Equations \eqref{eq:variation_step_1}-\eqref{eq:variations_step_5}, we have
    \begin{align*}
        \D \S^{\mathrm{red}}_X(x_{t}, u_{t}, y_{t}, \bar{\zeta}_t,\bar{\mu}_{t}) &= \int_0^T \left(\dotp{\frac{\partial \ell}{\partial x^{|T}_t}}{\delta x^{|T}_t} + \int_0^T \dotp{\frac{\partial \ell}{\partial u^{|T}_t}}{K_{\nabla^{T(Q/G)}}\delta u^{|T}_t} \right.\nonumber\\&+\left. \dotp{\frac{\partial \ell}{\partial \bar{\zeta}^{|T}_t}}{K_{\nabla^{T(Q/G)}}\delta \bar{\zeta}^{|T}_t}\right)\del X^0_t + \sum_{i = 1}^k  \int_0^T\dotp{\frac{\partial l_i}{\partial x_t^{|T}}}{\delta x_t^{|T}} \del X^i_t\\&+\int_0^T\dotp{K_{\nabla^{T^*(Q/G)}}\delta y^{|T}_t}{\del x^{|T}_t} - \int_0^T\dotp{\Del y^{|T}_t}{\delta x_t^{|T}}\\
        &-\int_0^T\dotp{K_{\nabla^{T^*(Q/G)}}\delta y^{|T}_t}{u_t^{|T}\del X^0_t + \sum_{i = 1}^k V_i^{\mathrm{red}}(x_t^{|T})\del X^i_t}\nonumber\\
      &- \int_0^T\dotp{ y^{|T}_t}{K_{\nabla^{T(Q/G)}}\delta u_t^{|T}}\del X^0_t- \sum_{i = 1}^k\int_0^T \dotp{\frac{\partial}{\partial x_t^{|T}}\dotp{y_t^{|T}}{V_i^{\reduced}(x_t^{|T})}}{\delta x_t^{|T}}\del X^i_t\\
      &+\int_0^T \dotp{K_{\nabla^{\dalag}}\delta\bar{\mu}^{|T}_t}{\del\bar{\xi}^{|T}_t}
-\int_0^T \dotp{\Del \bar{\mu}_t^{|T} -\ad^*_{\del\bar{\xi}^{|T}_t}\bar{\mu}^{|T}_t}{\bar{\eta}_t}\nonumber\\&-\int_0^T \dotp{\dotp{\bar{\mu}^{|T}_t}{i_{\del x^{|T}_t}\tilde{B}(x^{|T}_t)}}{\delta x^{|T}_t}-\int_0^T\dotp{K_{\nabla^{\dalag}}\delta\bar{\mu}_t^{|T}}{\bar{\zeta}_t^{|T}}\del X^0_t \nonumber\\&-\int_0^T\dotp{\bar{\mu}^{|T}_t}{K_{\nabla^{\alag}}\delta\bar{\zeta}_t^{|T}}\del X^0_t-\left(\sum_{i = 1}^k\int_0^T\dotp{K_{\nabla^{\dalag}}\delta\bar{\mu}_t^{|T}}{\bar{\beta}_i(x_t^{|T})}\del X^i_t\right.\\&+\left.\int_0^T \dotp{\frac{\partial}{\partial x_t^{|T}}\dotp{\bar{\mu}_t^{|T}}{\bar{\beta}_i(x_t^{|T})}}{\delta x_t^{|T}}\right).
    \end{align*}
    Then the expression in \eqref{eq:variationinSred_app} follows by appropriately grouping terms. 
    
    \medskip

     To prove the second part of the theorem, first note that if $(x^{|T}_t, u^{|T}_t, y^{|T}_t, \bar{\zeta}^{|T}_t, \bar{\mu}^{|T}_t)$ satisfies the horizontal and vertical stochastic implicit Lagrange-Poincaré equations then $ \D \S^{\mathrm{red}}_X(x_{t}, u_{t},y_{t}, \bar{\zeta}_{t},  \bar{\mu}_{t}) = 0$, so we only need to prove the converse. Take an arbitrary coordinate ball $K$ in $\P (Q/G)\oplus \alag\oplus \dalag$ and restrict to $(K,T)$-deformations of $(x^{|T}_t, u^{|T}_t, y^{|T}_t, \bar{\zeta}^{|T}_t, \bar{\mu}^{|T}_t)$, and semimartingales $\bar{\eta}_t$ that vanish outside $]]\tau_K^h, (\tau_K^h + \tau_K^{(h,e)})\wedge T[[$. Let $\tau_K^h$ denote the hitting time for $K$. Since $\bar{\eta}_t$ and the variations corresponding to $(K,T)$-deformations vanish outside $]]\tau_K^h, \tau_K^h + \tau_K^{(h,e)}[[$, the integral from $0$ to $T$ in the expression \eqref{eq:variationinSred_app} for $\D\S_X^{\reduced}(x_{t}, u_{t}, y_{t}, \bar{\zeta}_{t}, \bar{\mu}_{t})$ can be replaced by an integral from $\tau_K^h$ to $\tau_K^h + \tau_K^{(h,e)}$, that is,
     \begingroup
     \allowdisplaybreaks
     \begin{align*}
         \D\S_X^{\reduced}(x_{t}, u_{t}, y_{t}, \bar{\zeta}_{t}, \bar{\mu}_{t}) 
         &=
    \int_{\tau_K^h}^{\tau_K^h + \tau_{K}^{(h,e)}} \left\langle \frac{\partial}{\partial x_t^{|T}} \left( \ell \del X^0_t + \sum_{i =1}^k\left(l_i -  \dotp{y_t^{|T}}{V_i^{\reduced}(x_t^{|T})} - \dotp{\bar{\mu}_t^{|T}}{\bar{\beta}_i(x_t^{|T})}\right)\del X^i_t\right)\right.\nonumber\\
    &\left. - \dotp{\bar{\mu}_t}{i_{\del x_t}\tilde{B}(x_t)} - \Del y_t^{|T},\delta x_t^{|T}\right\rangle +\int_{\tau_K^h}^{\tau_K^h + \tau_{K}^{(h,e)}}\left\langle\frac{\partial \ell}{\partial u^{|T}_t} - y^{|T}_t, K_{\grad^{T(Q/G)}}\delta u^{|T}_t\right\rangle \del X^0_t\nonumber\\
    &+ \int_{\tau_K^h}^{\tau_K^h + \tau_{K}^{(h,e)}}\left\langle\frac{\partial \ell}{\partial {\bar{\zeta}}^{|T}_t} - \bar{\mu}^{|T}_t, K_{\grad^{\alag}}\delta {\bar{\zeta}}_t\right\rangle \del X^0_t\nonumber\\&+\int_{\tau_K^h}^{\tau_K^h + \tau_{K}^{(h,e)}} \langle K_{\grad^{T^*(Q/G)}} \delta y^{|T}_t, \del x^{|T}_t - u^{|T}_t\del X^0_t - \sum_{i = 1}^k V_i^{\mathrm{red}}(x^{|T}_t)\del X^i_t\rangle\nonumber\\ &+ \int_{\tau_K^h}^{\tau_K^h + \tau_{K}^{(h,e)}}\langle K_{\grad^{\dalag}}\delta \bar{\mu}^{|T}_t, \del \bar{\xi}^{|T}_t - \bar{\zeta}^{|T}_t\del X^0_t - \sum_{i = 1}^k\bar{\beta}_i(x^{|T}_t)\del X^i_t\rangle\nonumber\\&+\int_{\tau_K^h}^{\tau_K^h + \tau_{K}^{(h,e)}}\langle -\Del\bar{\mu}^{|T}_t+\ad^*_{\del \bar{\xi}^{|T}_t}\bar{\mu}^{|T}_t, \bar{\eta}_t\rangle
     \end{align*}
     \endgroup
     By the stochastic version of the fundamental lemma of the calculus of variations, Lemma \ref{lem:stochastic_fundlem}, this implies that the semimartingale $(x^{|T}_t, u^{|T}_t, y^{|T}_t, \bar{\zeta}^{|T}_t, \bar{\mu}^{|T}_t)$ satisfies the horizontal and vertical stochastic Lagrange-Poincaré equations in $]]\tau_K^{h}, \tau_K^h+\tau_K^{(h,e)}[[$. Since $K$ is an arbitrary regular coordinate ball in $\P (Q/G)\oplus \alag\oplus \dalag$, it follows that $(x^{|T}_t, u^{|T}_t, y^{|T}_t, \bar{\zeta}^{|T}_t, \bar{\mu}^{|T}_t)$ satisfies the horizontal and vertical Lagrange-Poincaré equations in $\P (Q/G)\oplus \alag\oplus\dalag$. This completes the proof.
\end{proof}

\end{document}